\documentclass[11pt,a4paper]{article}

\usepackage{amsmath,amsfonts,amssymb,amsthm}
\usepackage{mathbbol}
\usepackage{graphicx,color}
\usepackage{boxedminipage}
\usepackage[ruled, vlined, linesnumbered, boxed,noline]{algorithm2e}
\SetKw{Output}{output}
\usepackage{framed}
\usepackage{thmtools}
\usepackage{thm-restate}
\usepackage{xspace}
\usepackage{vmargin}
\setmarginsrb{1in}{1in}{1in}{1in}{0pt}{0pt}{0pt}{6mm}
\usepackage{todonotes}
 \usepackage[pdftex, plainpages = false, pdfpagelabels, 
                 bookmarks=false,
                 bookmarksopen = true,
                 bookmarksnumbered = true,
                 breaklinks = true,
                 linktocpage,
                 pagebackref,
                 colorlinks = true,  
                 linkcolor = blue,
                 urlcolor  = blue,
                 citecolor = red,
                 anchorcolor = green,
                 hyperindex = true,
                 hyperfigures
                 ]{hyperref} 
 \usepackage{xifthen}
 \usepackage{tabularx}

 \usetikzlibrary{calc}

\DeclareMathOperator{\operatorClassNP}{{\sf NP}}
\newcommand{\classNP}{\ensuremath{\operatorClassNP}}

\DeclareMathOperator{\operatorClassFPT}{{\sf FPT}\xspace}
\newcommand{\classFPT}{\ensuremath{\operatorClassFPT}\xspace}

\newcommand{\fp}{fully-polynomial}
\newcommand{\pd}{polynomial-delay}

\newcommand{\Oh}{\mathcal{O}}

\newtheorem{theorem}{Theorem}
\newtheorem{lemma}{Lemma}
\newtheorem{claim}{Claim}[theorem]
\newtheorem{corollary}{Corollary}
\newtheorem{definition}{Definition}
\newtheorem{observation}[theorem]{Observation}
\newtheorem{proposition}{Proposition}

\newtheorem{step}{Step}[theorem]
\newtheorem{reduction}[step]{Reduction Rule}
\newtheorem{branch}[step]{Branching Rule}

\newcommand{\pname}{\textsc}
\newcommand{\ProblemFormat}[1]{\pname{#1}}
\newcommand{\ProblemIndex}[1]{\index{problem!\ProblemFormat{#1}}}
\newcommand{\ProblemName}[1]{\ProblemFormat{#1}\ProblemIndex{#1}{}\xspace}

\newcommand{\probMC}{\ProblemName{Enum MC}}
\newcommand{\probMinMC}{\ProblemName{Enum Minimal MC}}
\newcommand{\probMaxMC}{\ProblemName{Enum Maximal MC}}

\newcommand{\NP}{{\ensuremath{\rm{NP}}}}
\newcommand{\coNP}{{\ensuremath{\rm{coNP}}}}

 \DeclareMathOperator{\poly}{poly}
\newcommand{\compass}{\NP\subseteq \coNP/\poly}
\newcommand{\ncompass}{\NP\nsubseteq \coNP/\poly}

\newcommand{\vc}{\tau}
\newcommand{\tw}{{\sf tw}}
\newcommand{\cw}{{\sf cw}}
\newcommand{\nd}{{\sf nd}}
\newcommand{\mw}{{\sf mw}}
\newcommand{\cp}{\theta}
\newcommand{\fen}{{\sf fn}}
\newcommand{\sol}{{\sf Sol}}
\newcommand{\mc}{\#_{\rm{mc}}}
\newcommand{\match}{\#_{\rm match}}
\newcommand{\tc}{{\sf tc}}



\newlength{\RoundedBoxWidth}
\newsavebox{\GrayRoundedBox}
\newenvironment{GrayBox}[1]%
   {\setlength{\RoundedBoxWidth}{.93\textwidth}
    \def\boxheading{#1}
    \begin{lrbox}{\GrayRoundedBox}
       \begin{minipage}{\RoundedBoxWidth}}%
   {   \end{minipage}
    \end{lrbox}
    \begin{center}
    \begin{tikzpicture}%
       \node(Text)[draw=black!20,fill=white,rounded corners,%
             inner sep=2ex,text width=\RoundedBoxWidth]%
             {\usebox{\GrayRoundedBox}};
        \coordinate(x) at (current bounding box.north west);
        \node [draw=white,rectangle,inner sep=3pt,anchor=north west,fill=white] 
        at ($(x)+(6pt,.75em)$) {\boxheading};
    \end{tikzpicture}
    \end{center}}     

\newenvironment{defproblemx}[2][]{\noindent\ignorespaces%
                                \FrameSep=6pt%
                                \parindent=0pt%
                \vspace*{-1.5em}
                \ifthenelse{\isempty{#1}}{%
                  \begin{GrayBox}{\textsc{#2}}%
                }{%
                  \begin{GrayBox}{\textsc{#2}  parameterized by~{#1}}%
                }
                \begin{tabular*}{\textwidth}{@{\hspace{.1em}} >{\itshape} p{1.8cm} p{0.8\textwidth} @{}}%
            }{
                \end{tabular*}%
                \end{GrayBox}%
                \ignorespacesafterend
            }  

\title{Refined Notions of Parameterized~Enumeration~Kernels with Applications to Matching~Cut~Enumeration\thanks{We dedicate this paper to the memory of our coauthor and friend Dieter Kratsch who recently passed away. Without Dieter, this paper would have never been written.}~\thanks{The research has been supported by the Research Council of Norway via the project ``MULTIVAL'' (grant no. 263317).}} 

\author{Petr A. Golovach\thanks{Department of Informatics, University of Bergen, Bergen, Norway}
\and
Christian Komusiewicz\thanks{Philipps-Universit\"at Marburg, Marburg, Germany}
\and
Dieter Kratsch\thanks{LGIMP, Universit\'e de Lorraine, Metz, France}
\and
Van Bang Le\thanks{Universit\"at Rostock, Institut f\"ur Informatik, Rostock, Germany}
}

\date{}

\begin{document}

\maketitle

\begin{abstract}
  An enumeration kernel as defined by Creignou et al.~[Theory Comput. Syst. 2017] for a
  parameterized enumeration problem consists of an algorithm that transforms each instance
  into one whose size is bounded by the parameter plus a solution-lifting algorithm that
  efficiently enumerates all solutions from the set of the solutions of the kernel. We
  propose to consider two new versions of enumeration kernels by asking that the solutions
  of the original instance can be enumerated in polynomial time or with polynomial delay
  from the kernel solutions. Using the NP-hard \textsc{Matching Cut} problem parameterized
  by structural parameters such as the vertex cover number or the cyclomatic number of the
  input graph, we show that the new enumeration kernels present a useful notion
  of data reduction for enumeration problems which allows to compactly represent the set
  of feasible solutions.
  
  
\end{abstract}

\section{Introduction}\label{sec:itro} 
The enumeration of all feasible solutions of a computational problem is a fundamental task
in computer science. 
For the majority
of enumeration problems, the number of feasible solutions can be exponential in the input size in the
worst-case. The running time of enumeration algorithms is thus measured not
only in terms of the input size~$n$ but also in terms of the output size. The two most-widely
used definitions of efficient algorithms are polynomial output-sensitive algorithms where
the running time is polynomial in terms of input and output size and polynomial-delay
algorithms, where the algorithm spends only a polynomial running time between the output
of consecutive solutions. Since in some enumeration problems, even the problem of deciding the existence of one solution is not solvable in polynomial time, it was proposed to allow \classFPT{} algorithms that have running time or delay~$f(k)\cdot n^{O(1)}$ for some problem-specific parameter~$k$~\cite{CreignouMMSV17,Damaschke06,Damaschke10,Fernau02,Meeks19}.
Naturally, \classFPT{}-enumeration algorithms are based on extensions of standard techniques in \classFPT{} algorithms such as bounded-depth search trees~\cite{Damaschke06,Damaschke10,Fernau02} or color coding~\cite{Meeks19}.

An important technique for obtaining \classFPT{} algorithms for decision problems is \emph{kernelization}~\cite{CyganFKLMPPS15,FominLSZ19,Kratsch14}, where the idea is to shrink the input instance in polynomial
time to an equivalent instance whose size depends only on the parameter~$k$. In fact, a
parameterized problem admits an \classFPT{} algorithm if and only if it admits a
kernelization. It seems particularly intriguing to use kernelization for enumeration problems as a small kernel can be seen as a compact representation of the set of feasible solutions.
The first notion of kernelization in the context of enumeration problems
were the \emph{full kernels} defined by Damaschke~\cite{Damaschke06}. Informally, a full kernel
for an instance of an enumeration problem is a subinstance that contains all minimal
solutions of size at most~$k$. This definition is somewhat restrictive since it is tied to
subset minimization problems parameterized by the solution size
parameter~$k$. Nevertheless, full kernels have been obtained for some 
problems~\cite{Damaschke10,FominSV13,KU14,SS08}.

To overcome the restrictions of full kernels, Creignou et
al.~\cite{CreignouMMSV17} proposed \emph{enumeration kernels}. Informally, an enumeration
kernel for a parameterized enumeration problem is an algorithm that replaces the input
instance by one whose size is bounded by the parameter and which has the property that the
solutions of the original instance can be computed by listing the solutions of the kernel
and using an efficient \emph{solution-lifting algorithm} that outputs for each solution of the
kernel a set of solutions of the original instance. In the definition of Creignou et
al.~\cite{CreignouMMSV17}, the solution-lifting algorithm may be an \classFPT{}-delay
algorithm, that is, an algorithm with~$f(k)\cdot n^{O(1)}$ delay where~$n$ is the
overall input size.
We find that this time bound is too weak, because it essentially implies that every enumeration problem that can be solved with \classFPT{}-delay admits an enumeration kernel of \emph{constant} size. Essentially, this means that the solution-lifting algorithm is so powerful that it can enumerate all solutions while ignoring the kernel. Motivated by this observation and the view of kernels as compact representations of the solution set, we modify the original definition of enumeration kernels~\cite{CreignouMMSV17}. 

\paragraph{Our results.}  We present two new notions of
efficient enumeration kernels by replacing the demand for \classFPT{}-delay algorithms by
a demand for polynomial-time enumeration algorithms or polynomial-delay algorithms,
respectively. We call the two resulting notions of enumeration kernelization
\emph{fully-polynomial enumeration kernels} and \emph{polynomial-delay enumeration
  kernels}. Our paper aims at showing that these two new definitions present a sweet spot
between the notion of full kernels, which is too strict for some applications, and
enumeration kernels, which are too lenient in some sense. We first show that the two new
definitions capture the class of efficiently enumerable problems in the sense that a
problem has a fully-polynomial (a polynomial-delay) enumeration kernel if and only if it
has an \classFPT{}-enumeration algorithm (an \classFPT{}-delay enumeration
algorithm). Moreover, the kernels have constant size if and only if the problems have
polynomial-time (polynomial-delay) enumeration algorithms. Thus, the new definitions
correspond to the case of problem kernels for decision problems, which are in \classFPT{} if and
only if they have kernels and which can be solved in polynomial time if and only if they have 
kernels of constant size (see, e.g.~\cite[Chapter~2]{CyganFKLMPPS15} or \cite[Chapter~1]{FominLSZ19}).

  We then apply both types of kernelizations to the enumeration of matching
  cuts. A matching cut of a graph $G$ is the set of edges $M=E(A,B)$ for  a partition $\{A,B\}$ of $V(G)$ forming a matching.  
  We
  investigate the problems of enumerating all minimal,  all maximal, or all matching cuts
  of a graph.  We refer to these problems as \probMinMC, \probMaxMC, and~\probMC,
  respectively. These matching cut problems constitute a very suitable study case for enumeration kernels, since it is NP-hard to decide whether a graph has a matching cut~\cite{Chvatal84} and therefore, they do not admit polynomial output-sensitive algorithms. We consider all three problems with respect to structural
  parameterizations such as the vertex cover number, the modular width, or the cyclomatic
  number of the input graph. 
  The choice of these parameters is motivated by the fact that neither problem admits an enumeration kernel of polynomial size for the more general structural parameterizations by the treewidth or cliquewidth  up to some natural complexity assumptions (see Proposition~\ref{prop:no-kern-tw}).
  Table~\ref{tab:overview} summarizes the results.
 
  To discuss some of our results and their implication for enumeration kernels in general
  more precisely, consider \probMC, \probMinMC, and \probMaxMC
  parameterized by the vertex cover number. We show that \probMinMC{} admits a
  fully-polynomial enumeration kernel of polynomial size. As it can be seen that the problem has no full kernel, in particular, 
  this implies that
  there are natural enumeration problems with a fully-polynomial enumeration kernel that
  do not admit a full kernel (not even one of super-polynomial
  size). Then, we show that \probMC and \probMaxMC admit polynomial-delay enumeration
  kernels but have no fully-polynomial enumeration kernels. Thus, there are natural enumeration problems with 
  polynomial-delay enumeration kernels that do not admit fully-polynomial enumeration
  kernels (not even one of super-polynomial size).

\begin{table}[ht]
  \caption{An overview of our results. Herein, \lq kernel\rq\ means \fp{} enumeration kernel, \lq delay-kernel\rq\ means \pd{} enumeration kernel and \lq bi\rq\ means bijective enumeration kernel (a slight generalization of full kernels), a ($\star$) means that the lower bound assumes $\ncompass$, \lq ?\rq\ means open status. 
  The cyclomatic number is also known as the feedback edge number.}\label{tab:overview}
   \begin{center}
{ \small  
\begin{tabular}{l|l|l|l}
Parameter $k$\,  & \probMC      & \probMinMC      & \probMaxMC \\
\hline 
treewidth \&  & No poly-size delay-  & No poly-size delay- & No poly-size delay-  \\
cliquewidth & kernel ($\star$) (Prop.~\ref{prop:no-kern-tw})      & kernel ($\star$) (Prop.~\ref{prop:no-kern-tw})     & kernel ($\star$) (Prop.~\ref{prop:no-kern-tw}) \\
\hline 
vertex cover \& & size-$\Oh(k^2)$ delay-kernel   & size-$\Oh(k^2)$ kernel & size-$\Oh(k^2)$ delay-kernel\\
twin-cover    & (Theorems~\ref{thm:vc-kern} \& \ref{thm:tc-kern}) & (Theorems~\ref{thm:vc-kern} \& \ref{thm:tc-kern}) &(Theorems~\ref{thm:vc-kern} \& \ref{thm:tc-kern})\\
number        & No kernel                     &  & No kernel\\
\hline
neighborhood & size-$\Oh(k)$ delay-kernel    & size-$\Oh(k)$ kernel        & size-$\Oh(k)$ delay-kernel \\
diversity    & (Theorem~\ref{thm:nd-kern})& (Theorem~\ref{thm:nd-kern}) & (Theorem~\ref{thm:nd-kern})\\
               & No kernel                &                             & No kernel \\
\hline
modular  & ?     & $\Oh(k)$-kernel             & ?  \\
width   &  & (Theorem~\ref{thm:mw-kern}) & \\    
\hline
cyclomatic & size-$\Oh(k)$ delay-kernel        & size-$\Oh(k)$ delay-kernel & ?\\
number     & (Theorem~\ref{thm:fen-kern-min}) & (Theorem~\ref{thm:fen-kern-min}) & \\
           & No kernel                   &   &\\
\hline
clique partition & size-$\Oh(k^3)$ bi kernel  & size-$\Oh(k^3)$ bi kernel &size-$\Oh(k^3)$ bi kernel \\
number           &(Theorem~\ref{thm:cp}) & (Theorem~\ref{thm:cp}) & (Theorem~\ref{thm:cp})
\end{tabular}
}
\end{center}
\end{table}

We also prove a tight upper bound $F(n + 1)-1$ for the maximum number of matching cuts of an $n$-vertex graph, where $F(n)$ is the $n$-th Fibonacci number and show that all matching cuts can be enumerated in $\Oh^*(F(n))=\Oh^*(1.6181^n)$ time (Theorem~\ref{thm:general}). 

\paragraph{Related work.}
  The current-best exact decision algorithm for \textsc{Matching Cut}, the problem of deciding whether a given graph~$G$ has a matching cut, has a running time of~$O(1.328^n)$ where~$n$ is the number of vertices in~$G$~\cite{KomusiewiczKL20}. Faster exact algorithms can be obtained for the case when the minimum degree is large~\cite{HsiehLLP19}. \textsc{Matching Cut} has \classFPT{}-algorithms for the maximum cut size~$k$~\cite{GomesS19}, the vertex cover number of~$G$~\cite{KratschL16}, and weaker parameters such as the twin-cover number~\cite{AravindKK17} or the cluster vertex deletion number~\cite{KomusiewiczKL20}. 

  For an overview of enumeration 
  algorithms, refer to the survey of
  Wasa~\cite{Wasa16}. A broader discussion of parameterized enumeration is given by
  Meier~\cite{Meier20}.  A different extension of enumeration kernels are \emph{advice
    enumeration kernels}~\cite{BentertNN19}. In these kernels, the solution-lifting
  algorithm does not need the whole input 
    but only a possibly smaller advice. A further loosely connected extension of standard kernelization are \emph{lossy kernels} which are used for optimization problems~\cite{LokshtanovPRS17}; the common thread is that both definitions use a solution-lifting algorithm for recovering solutions of the original instance.

\paragraph*{Graph notation.} All graphs considered in this paper are finite undirected graphs without loops or multiple edges.   
We follow the standard graph-theoretic notation and terminology and refer to the book of Diestel~\cite{Diestel12} for basic definitions. For each of the graph problems considered in this paper, we let $n=|V(G)|$ and $m=|E(G)|$ denote the number of vertices and edges, respectively, of the input graph $G$ if it does not create confusion. 
For a graph~$G$ and a subset $X\subseteq V(G)$ of vertices, we write $G[X]$ to denote the subgraph of $G$ induced by~$X$.
For a set of vertices $X$, $G-X$ denotes the graph obtained by deleting the vertices of~$X$, that is, $G-X=G[V(G)\setminus X]$; for a vertex $v$, we write $G-v$ instead of $G-\{v\}$.
Similarly, for a set of edges $A$ (an edge $e$, respectively), $G-A$ ($G-e$, respectively) denotes the graph obtained by the deletion of the edges of $A$ (the edge $e$, respectively). 
For a vertex $v$, we denote by $N_G(v)$ the \emph{(open) neighborhood} of $v$, i.e., the set of vertices that are adjacent to $v$ in $G$.
We use $N_G[v]$ to denote the \emph{closed neighborhood} $N_G(v)\cup \{v\}$ of~$v$.
For $X\subseteq V(G)$, $N_G[X]=\bigcup_{v\in X}N_G[v]$ and $N_G(X)=N_G[X]\setminus X$. 
For disjoint sets of vertices $A$ and $B$ of a graph $G$, $E_G(A,B)=\{uv\mid u\in A,~v\in B\}$. We may omit subscripts in the above notation if it does not create confusion. We use $P_n$, $C_n$, and $K_n$ to denote the $n$-vertex path, cycle, and complete graph, respectively.  We write $G+H$ to denote the \emph{disjoint union} of $G$ and $H$, and we use $kG$ to denote the disjoint union of $k$ copies of $G$. 

In a graph $G$, a \emph{cut} is a partition $\{A,B\}$ of $V(G)$, and we say that $E_G(A,B)$ is an \emph{edge cut}. 
 A \emph{matching} is an edge set in which no two of the edges have a common end-vertex; note that we allow empty matchings. 
A \emph{matching cut} is a (possibly empty) edge set being an
edge cut and a matching. We underline that by our definition, a matching cut is a set of edges, as sometimes in the literature  (see, e.g.,~\cite{Chvatal84,Graham70}) a matching cut is defined as a partition $\{A,B\}$ of the vertex set such that $E(A,B)$ is a matching. 
While the two variants of the definitions  are equivalent, say when the decision variant of the matching cut problem is considered, this 
is not the case in enumeration and counting when we deal with disconnected graphs. 
For example, the empty graph on $n$ vertices has $2^{n-1}-1$ partitions 
$\{A,B\}$ which all correspond to exactly one matching cut in the sense of our definition, namely $M=\emptyset$.
A matching cut~$M$ of $G$ is \emph{(inclusion) minimal} (\emph{maximal}, respectively) if $G$ has no matching cut $M'\subset M$ ($M'\supset M$, respectively). Notice that a disconnected graph has exactly one minimal matching cut which is the empty set. 

\paragraph*{Organization of the paper.} In Section~\ref{sec:basics}, we introduce and discuss basic notions of enumeration kernelization. In Section~\ref{sec:upper}, we show upper and lower bound for the maximum number of (minimal) matching cuts. In Section~\ref{sec:vc}, we give enumeration kernels for the matching cut problems parameterized by the vertex cover number. Further, in Section~\ref{sec:nd}, we consider the parameterization by the neighborhood diversity and modular width. We proceed in Section~\ref{sec:fen}, by investigating the parameterization by the cyclomatic number (feedback edge number). In Section~\ref{sec:cp}, we give bijective kernels for the parameterization by the clique partition number.  We conclude in Section~\ref{sec:concl}, by outlining some further directions of research in enumeration kernelization.

\section{Parameterized Enumeration and Enumeration Kernels}\label{sec:basics}
We  use the framework for parameterized enumeration proposed by Creignou et al.~\cite{CreignouMMSV17}.
An \emph{enumeration problem} (over a finite alphabet $\Sigma$) is a tuple $\Pi=(L,\sol)$ such that 
\begin{itemize}
\item[(i)] $L\subseteq \Sigma^*$ is a decidable language,
\item[(ii)] $\sol\colon \Sigma^*\rightarrow \mathcal{P}(\Sigma^*)$ is a computable function such that for every $x\in \Sigma^*$, $\sol(x)$ is a finite set and $\sol(x)\neq\emptyset$ if and only if $x\in L$.
\end{itemize}
Here, $\mathcal{P}(A)$ is used to denote the powerset of a set~$A$. A string $x\in \Sigma^*$ is an \emph{instance},  and $\sol(x)$ is the set of solutions to instance $x$. 
A \emph{parameterized enumeration problem} is defined as a triple $\Pi=(L,\sol,\kappa)$ such that $(L,\sol)$ satisfy (i) and (ii) of the above definition, and 
\begin{itemize}
\item[(iii)] $\kappa\colon \Sigma^*\rightarrow \mathbb{N}$ is a parameterization.
\end{itemize}
We say that  $k=\kappa(x)$ is a \emph{parameter}. We define the parameterization as a function of an instance but it is standard to assume that the value of $\kappa(x)$ is either simply given in $x$ or can be computed in polynomial time from $x$. We follow this convention throughout the paper.

An \emph{enumeration algorithm} $\mathcal{A}$ for a parameterized enumeration problem $\Pi$ is a deterministic algorithm that for every instance~$x$, outputs exactly the elements of $\sol(x)$ without duplicates, and terminates after a finite number of steps on every instance.  
The algorithm $\mathcal{A}$ is an \emph{\classFPT enumeration} algorithm if it outputs all solutions in at most $f(\kappa(x))p(|x|)$ steps for a computable function $f(\cdot)$ that depends only on the parameter and a polynomial $p(\cdot)$. 

We also consider output-sensitive enumerations, and for this, we define delays.
Let $\mathcal{A}$ be an enumeration algorithm for $\Pi$. For $x\in L$ and $1\leq i\leq |\sol(x)|$, the \emph{$i$-th delay} of $\mathcal{A}$ is the time between outputting the $i$-th and $(i+1)$-th solutions in $\sol(x)$.  The $0$-th delay is the \emph{precalculation} time which is the time from the start of the computation until the output of the fist solution, and the $|\sol(x)|$-th delay is the \emph{postcalculation} time which is the time after the last output and the termination of $\mathcal{A}$ (if $\sol(x)=\emptyset$, then the precalculation and postcalculation times are the same).
It is said that $\mathcal{A}$ is a \emph{polynomial-delay} algorithm, if all the delays are upper bounded by $p(|x|)$ for a polynomial~$p(\cdot)$. For  a parameterized enumeration  problem $\Pi$, $\mathcal{A}$ is an \emph{\classFPT-delay algorithm}, if the delays are at most $f(\kappa(x))p(|x|)$, where $f(\cdot)$ is a computable function and $p(\cdot)$ is a polynomial. Notice that every \classFPT enumeration algorithm $\mathcal{A}$ is also an \classFPT delay algorithm.  

The key definition for us is the generalization of the standard notion of a kernel in Parameterized Complexity (see, e.g,~\cite{FominLSZ19}) for enumeration problems.
\begin{definition}\label{def:enum-kernel}
Let $\Pi=(L,\sol,\kappa)$ be a parameterized enumeration problem. A \emph{fully-polynomial enumeration kernel(ization)}  for $\Pi$ is a pair of algorithms $\mathcal{A}$ and $\mathcal{A}'$ with the following properties:
\begin{itemize}
\item[\em (i)] For every instance $x$ of $\Pi$, $\mathcal{A}$ computes in time polynomial in $|x|+\kappa(x)$  an instance $y$ of~$\Pi$ such that 
$|y|+\kappa(y)\leq f(\kappa(x))$ for a computable function $f(\cdot)$.
\item[\em (ii)] For every $s\in\sol(y)$, $\mathcal{A}'$ computes in time polynomial in $|x|+|y|+\kappa(x)+\kappa(y)$  a nonempty set of solutions $S_s\subseteq \sol(x)$ such that $\{S_s\mid s\in \sol(y)\}$ is a partition of $\sol(x)$.
\end{itemize}
\end{definition}
Notice that by (ii), $x\in L$ if and only if $y\in L$.

We say that $\mathcal{A}$ is a \emph{kernelization} algorithm and $\mathcal{A}'$ is a \emph{solution-lifting} algorithm. Informally, a solution-lifting algorithm takes as its input a solution for a ``small'' instance constructed by the kernelization algorithm and, having an access to the original input instance,  
outputs polynomially many solutions for the original instance, and  by going over all the solutions to  the small instance, we can generate all the solutions of the original instance without repetitions. 
We say that an enumeration kernel is \emph{bijective} if $\mathcal{A}'$ produces a unique solution to $x$, that is, it establishes a bijection between $\sol(y)$ and $\sol(x)$, that is, the compressed instance essentially has the same solutions as the input instance. In particular, full kernels~\cite{Damaschke06} are the special case of bijective kernels where~$\mathcal{A}'$ is the identity.   
As it is standard, $f(\cdot)$ is the \emph{size} of the kernel, and the kernel has \emph{polynomial size} if $f(\cdot)$ is a polynomial. 

We define   
\emph{polynomial-delay enumeration kernel(ization)}
in a similar way. The only difference is that (ii) is replaced by the condition
\begin{itemize} 
\item[(ii$^*$)] For every $s\in\sol(y)$, $\mathcal{A}'$ computes with delay polynomial in $|x|+|y|+\kappa(x)+\kappa(y)$  a set of solutions $S_s\subseteq \sol(x)$ such that $\{S_s\mid s\in \sol(y)\}$ is a partition of $\sol(x)$.
\end{itemize}
It is straightforward to make the following observation.

\begin{observation}\label{obs:input-special}
Every bijective enumeration kernel is a \fp{} enumeration kernel; every  \fp{} enumeration kernel is a \pd{} enumeration 
kernel.
\end{observation}

Notice also that our definition of 
polynomial-delay enumeration 
kernel is different from the definition given by Creignou et al.~\cite{CreignouMMSV17}.
In their definition, Creignou et al.~\cite{CreignouMMSV17} require that the solution-lifting algorithm $\mathcal{A}'$ should list all the solutions in $S_s$ with \classFPT delay for the parameter $\kappa(x)$.  We believe that this condition is too weak.  In particular, with this requirement, every parameterized enumeration problem, that 
has an \classFPT enumeration algorithm $\mathcal{A}^*$ and such that the existence of at least one solution can be verified in polynomial time, has a trivial kernel of constant size. 
The kernelization algorithm can output any instance satisfying (i) and then we can use $\mathcal{A}^*$ as a solution-lifting algorithm that essentially ignores the output of the kernelization algorithm.
Note that for enumeration problems, we typically face the situation where the existence of at least one solution is not an issue.
We argue that our definitions are natural by showing  the following theorem.

\begin{theorem}\label{obs:natural}
A parameterized enumeration problem~$\Pi$ has an \classFPT enumeration algorithm (an \classFPT delay algorithm) if and only if $\Pi$ admits a fully-polynomial enumeration kernel (polynomial-delay enumeration kernel). Moreover, $\Pi$ can be solved in polynomial time (with polynomial delay) if and only if $\Pi$ admits a fully-polynomial enumeration kernel 
(a~polynomial-delay enumeration
kernel) of constant size. 
\end{theorem}

\begin{proof}
The proof of the first claim is similar to the standard arguments for showing the equivalence between fixed-parameter tractability and the existence of a kernel  (see, e.g.~\cite[Chapter~2]{CyganFKLMPPS15} or \cite[Chapter~1]{FominLSZ19}). However dealing with enumeration problems requires some specific arguments.
Let $\Pi=(L,\sol,\kappa)$
be a parameterized enumeration problem.

In the forward direction, the claim is trivial. Recall that $L$ is decidable and $\sol(\cdot)$ is a computable function by the definition. If $\Pi$ admits a fully-polynomial enumeration kernel (a polynomial-delay enumeration kernel
respectively), then we apply an arbitrary enumeration algorithm, which is known to exist since $\sol(\cdot)$ is computable, to the instance $y$ produced by the kernelization algorithm. Then, for each $s\in \sol(y)$, use the solution-lifting algorithm to list the solutions to the input instance.

For the opposite direction, assume that $\Pi$ can be solved in $f(\kappa(x))\cdot |x|^c$ time (with $f(\kappa(x))\cdot |x|^c$ delay, respectively) for an instance~$x$, where $f(\cdot)$ is a computable function and $c$ is a positive constant. Since $f(\cdot)$ is computable, we assume that we have an algorithm $\mathcal{F}$ computing $f(k)$ in $g(k)$ time. We define $h(k)=\max\{f(k),g(k)\}$. 

We say that an instance $x$ of $\Pi$ is a \emph{trivial no-instance} if $x$ is an instance of minimum size with $\sol(x)=\emptyset$.
We call $x$ a \emph{minimum yes-instance} if $x$ is an instance of minimum size that has a solution. Notice that if $\Pi$ has instances without solutions, then the size of a trivial no-instance is a constant that depends on $\Pi$ only and such an instance can be computed in constant time. Similarly, if the problem has instances with solutions, then the size of a minimum yes-instance is constant and such an instance can be computed in constant time. 
We say that $x$ is a \emph{trivial yes-instance} if $x$ is an instance with minimum size of $\sol(x)$ that, subject to the first condition, has minimum size. Clearly, the size of a trivial yes-instance is a constant that depends only on~$\Pi$. However, we may be unable to compute a trivial yes-instance. 

Let $x$ be an instance of $\Pi$ and $k=\kappa(x)$. We run the algorithm $\mathcal{F}$ to compute $f(k)$ for at most $n=|x|$ steps. If the algorithm failed to compute $f(k)$ in $n$ steps, we conclude that $g(k)\geq n$.
In this case, the kernelization algorithm outputs $x$. Then the solution-lifting algorithm just trivially outputs its input solutions.  Notice that $|x|\leq g(k)\leq h(k)$ in this case.
Assume from now that $\mathcal{F}$ computed $f(k)$ in at most $n$ steps. 

If $|x|\le f(k)$, then the kernelization algorithm outputs the original instance $x$, and the solution-lifting algorithm trivially outputs its input solutions. Note that $|x|\leq f(k)\leq h(k)$.

Finally, we suppose that $f(k)<|x|$. Observe that the enumeration algorithm runs in $|x|^{c+1}$ time (with $|x|^{c+1}$ delay, respectively) in this case, that is, the running time is polynomial. We use the enumeration algorithm to verify whether $x$ has a solution. For this, notice that a polynomial-delay algorithm can be used to solve the decision problem; we just run it  until it outputs a first solution (or reports that there are no solutions). 
If $x$ has no solution, then $\Pi$ has a trivial no-instance and the kernelization algorithm computes and outputs it. If $x$ has a solution, then the kernelization algorithm computes a minimum yes-instance $y$ in constant time. We use the enumeration algorithm to check whether $|\sol(y)|\leq  |\sol(x)|$. If this holds, then we set $z=y$. Otherwise, if $|\sol(x)|<|\sol(y)|$, we find an instance $z$ of minimum size such that $|\sol(z)|\leq |\sol(x)|$. Notice that this can be done in constant time, because the size of $z$ is upper bounded by the size of a trivial yes-instance. Then we list the solutions of $z$ in constant time and order them. For the $i$-th solution of $z$, the solution-lifting algorithm outputs the $i$-th solution of $x$ produced by the enumeration algorithm, and for the last solution of $z$, the solution-lifting algorithm further runs the enumeration algorithm to output the remaining solutions. Since $|\sol(z)|\leq|\sol(x)|$, the solution-lifting algorithm outputs a nonempty set of solutions for $x$ for every solution of $z$. 
 
 It is easy to see that we obtain a fully-polynomial enumeration kernel of size $\Oh(h(\kappa(x))$ (a~polynomial-delay enumeration
kernel, respectively).

For the second claim, the arguments are the same. 
If a problem admits a 
fully-polynomial (a polynomial-delay) enumeration kernel of constant size, then the solutions of the original instance can be listed in polynomial time (or with polynomial delay, respectively) by the solution-lifting algorithm called for the constant number of the solutions of the kernel. Conversely, if a problem can be solved in polynomial time (with polynomial delay, respectively), we can apply the above arguments assuming that $f(k)$ (and, therefore, $g(k)$) is a constant.
\end{proof}

In our paper, we consider structural parameterizations of \probMinMC, \probMaxMC, and \probMC by several graph parameters, and the majority of these parameterizations are stronger than the parameterization either by the  \emph{treewidth} or the \emph{cliquewidth} of the input graph.  Defining the treewidth (denoted by $\tw(G)$) and cliquewidth (denoted by~$\cw(G)$) goes beyond 
of the scope of the current paper and we refer to~\cite{Courcelle09}  (see also, e.g.,~\cite{CyganFKLMPPS15}). 
By the celebrated result of Bodlaender~\cite{Bodlaender96} (see also~\cite{CyganFKLMPPS15}), it is  \classFPT in $t$ to decide whether $\tw(G)\leq t$ and to construct the corresponding  tree-decomposition. No such  algorithm is known for cliquewidth. However, for algorithmic purposes, it is usually sufficient to use the approximation algorithm of Oum and Seymour~\cite{OumS06} (see also~\cite{Oum08,CyganFKLMPPS15}).  
Observe that the property that a set of edges $M$ of a graph  $G$ is a matching cut of $G$ can be expressed in  \emph{monadic second-order logic} (MSOL);
we refer to~\cite{Courcelle09,CyganFKLMPPS15} for the definition of MSOL on graphs. 
Then the  matching cuts (the minimal or maximal matching cuts) of a graph of treewidth at most $t$ can be enumerated with   \classFPT delay with respect to the parameter $t$ by the celebrated meta theorem of Courcelle~\cite{Courcelle09}. The same holds for the weaker parameterization by the cliquewidth  of the input graph, because we can use MSOL formulas without quantifications over (sets of) edges: For a graph $G$, we pick a vertex in each connected component of $G$ and label it. Let $R$ be the set of labeled vertices. Then the enumeration of nonempty matching cuts is equivalent to the enumeration of all partitions $\{A,B\}$ of $V(G)$ such that (i) $R\subseteq A$ and (ii) $E(A,B)$ is a matching. Notice that condition (ii) can be written as follows: for every $u_1,u_2\in A$ and $v_1,v_2\in B$, if $u_1$ is adjacent to $v_1$ and $u_2$ is adjacent to $v_2$, then either $u_1=u_2$ and $v_1=v_2$ or $u_1\neq u_2$ and $v_1\neq v_2$. Since the empty matching cut can be listed separately if it exists,
we obtain that we can use MSOL formulations of the enumeration problems, where only quantifications over vertices and sets of vertices are used. Then the result of Courcelle~\cite{Courcelle09} implies  that \probMinMC, \probMaxMC, and \probMC  can be solved with  \classFPT delay when parameterized by the cliquewidth of the input graph. 
We summarize these observations in the following proposition.
\begin{proposition}\label{prop:tw-cw}
\probMC, \probMinMC, and \probMaxMC on  graphs of treewidth (cliquewidth) at most $t$ can be solved with \classFPT delay when parameterized by $t$.  
\end{proposition}
This proposition implies that \probMC, \probMinMC and \probMaxMC can be solved with \classFPT delay for all structural parameters whose values can be bounded from below by an increasing function of treewidth or cliquewidth. However, we are mainly interested in 
fully-polynomial or polynomial-delay enumeration kernelization.  
 We conclude this section by pointing out that it is unlikely that \probMinMC, 
\probMaxMC, and \probMC admit polynomial-delay 
enumeration kernels of polynomial size
for the treewidth or cliquewidth parameterizations. It was pointed out by  Komusiewicz, Kratsch, and Le~\cite{KomusiewiczKL20} that the decision version of the matching cut problem (that is, the problem asking whether a given graph~$G$ has a matching cut) does not admit a polynomial kernel when parameterized by the treewidth of the input graph unless $\compass$. By the definition of 
a \pd{} enumeration
kernel, this gives the following statement.
\begin{proposition}\label{prop:no-kern-tw}
 \probMinMC, \probMaxMC and \probMC do not admit polynomial-delay enumeration  
 kernels of polynomial size
  when parameterized by the treewidth (cliquewidth, respectively) of the input graph unless $\compass$. 
 \end{proposition}
\section{A Tight Upper Bound for the Maximum Number of Matching Cuts}\label{sec:upper}
In this section we provide a tight upper bound for the maximum number of matching cuts of an $n$-vertex graph. We complement this result by giving an exact enumeration algorithm for (minimal, maximal) matching cuts. Finally, we give some lower bounds for the maximum number of minimal and maximal matching cuts, respectively. Throughout this section, we use $\mc(G)$ to denote the number of matching cuts of a graph $G$.

To give the upper bound, we use the classical Fibonacci numbers. For a positive integer~$n$, we denote by $F(n)$ the $n$-th Fibonacci number. Recall that $F(1)=F(2)=1$, and for $n\geq 3$, the Fibonacci numbers satisfy the recurrence $F(n)=F(n-1)+F(n-2)$. Recall also that the $n$-th Fibonacci number can be expressed by the following closed formula:
 \begin{equation*}
F(n)=\frac{1}{\sqrt{5}}\Big(\Big(\frac{1+\sqrt{5}}{2}\Big)^n+\Big(\frac{1-\sqrt{5}}{2}\Big)^n\Big)
\end{equation*}
for every $n\geq 1$. In particular, $F(n)=\Oh(1.6181^n)$. 

The following lemma about the Fibonacci numbers is going to be useful for us.

\begin{lemma}\label{lem:fib}
For all integers $p,q\geq 2$, $F(p)F(q)\leq F(p+q-1)-1$. Moreover, if $p\geq 4$ or $q\geq 4$, then the inequality is strict.
\end{lemma}

\begin{proof}
The proof is inductive. It is straightforward to verify the inequality for $p,q\leq 3$. Notice that $F(p)F(q)= F(p+q-1)-1$ in these cases. Assume now that $p\geq 4$. Then, by induction,
\begin{align*}
F(p)F(q)=&F(p-1)F(q)+F(p-2)F(q)\\
\leq& F(p+q-2)-1+F(p+q-3)-1=F(p+q-1)-2\\
<&F(p+q-1)-1,
\end{align*}
as required. 
\end{proof}

To see the relations between the number of matching cuts and the Fibonacci number, we make the following observation. 

\begin{observation}\label{obs:path}
For every positive integer $n$, the $n$-vertex path has $F(n+1)-1$ matching cuts. 
\end{observation}

\begin{proof}
The proof is by induction. Clearly, $\mc(P_1)=0=F(2)-1$ and $\mc(P_2)=1=F(3)-1$. Let $n\geq 3$ and $P=v_1\cdots v_n$. Then $M\subseteq E(P)$ is a matching cut of $P$ if and only if $M$ is a nonempty matching and either $M=\{v_1v_2\}$, or $M$ is a nonempty matching of $P'=v_2\cdots v_n$, or $M=M'\cup\{v_1v_2\}$, where $M'$ is a nonempty matching of $P''=v_3\cdots v_n$.  Therefore,
\begin{equation*}
\mc(P_n)=\mc(P)=1+\mc(P')+\mc(P'')=1+\mc(P_{n-1})+\mc(P_{n-2}).
\end{equation*}
By induction, we conclude that
\begin{equation*}
\mc(P_n)=1+\mc(P_{n-1})+\mc(P_{n-2})=1+(F(n)-1)+(F(n-1)-1)=F(n+1)-1
\end{equation*}
as required. 
\end{proof}

We show that, in fact, $F(n+1)-1$ is an upper bound for the number of matching cuts of an $n$-vertex graph. First, we show this for trees.

\begin{lemma}\label{lem:tree}
An $n$-vertex tree $T$ has at most $F(n+1)-1$ matching cuts. Moreover, the bound  $F(n+1)-1$ is tight and is achieved if and only if $T$ is a path.  
\end{lemma}

\begin{proof}
Clearly, $M\subseteq E(T)$ is a matching cut of $T$ if and only if $M$ is a nonempty matching. Denote by $\match(G)$ the number of nonempty matchings of a graph $G$. We show by induction that for every $n$-vertex forest $H$, $\match(H)\leq F(n+1)-1$ and the inequality is strict whenever $H$ is not a path. The claim is straightforward  if $n\leq 2$. Let $n\geq 3$. If $H$ has no edges, $\match(H)=0$ and the claim holds. Otherwise, $H$ has a leaf $u$. Denote by $v$ the unique neighbor of $u$. Clearly,  $M$ is a nonempty matching of $H$ if and only if either $M=\{uv\}$, or $M$ is a nonempty matching of $H'=H-u$, or $M=M'\cup\{uv\}$, where $M'$ is a nonempty matching of $H''=H-\{u,v\}$.
We have that
\begin{equation*}
\match(H)=1+\match(H-u)+\match(H-\{u,v\}).
\end{equation*} 
By induction, 
\begin{equation}\label{eq:match} 
\match(H)\leq 1+(F(n)-1)+(F(n-1)-1)=F(n+1)-1.
\end{equation}
To see the second claim, note that if $H$ is not a path, then either $H-u$ or $H-\{u,v\}$ is not a path. Then, by induction, the inequality in (\ref{eq:match}) is strict and, therefore, $\match(H)<F(n+1)-1$. This concludes the proof.
\end{proof}

It is well-known that the treewidth of a tree is one (see, e.g.,~\cite{CyganFKLMPPS15}). This observation together with Proposition~\ref{prop:tw-cw} and Lemma~\ref{lem:tree} immediately imply the following lemma.

\begin{lemma}\label{lem:enum-tree} 
The matching cuts of an $n$-vertex tree can be enumerated in  $\Oh^*(F(n))$~time.
\end{lemma}

It is also easy to construct a direct  enumeration algorithm for trees. For example, one can consider the recursive branching algorithm that for an edge, first enumerates matching cuts containing this edge and then the matching cuts excluding the edge. Note that the running time in Lemma~\ref{lem:enum-tree} can be written as $\Oh(1.6181^n)$ to make the exponential dependence on $n$ more clear.

Now we consider general graphs and show the following.

\begin{theorem}\label{thm:general}
An $n$-vertex graph has at most $F(n+1)-1$ matching cuts. The bound is tight and is achieved for paths. Moreover, if $n\geq 5$, then an $n$-vertex graph $G$ has $F(n+1)-1$ matching cuts if and only if $G$ is a path. Furthermore, the matching cuts can be enumerated in  $\Oh^*(F(n))$ time.
\end{theorem}

\begin{proof}
First, we consider connected graphs.

 Let $G$ be a connected graph. 
Observe that if $M$ is a matching cut of $G$, then the partition $\{A,B\}$ of $V(G)$ such that $M=E(A,B)$ is unique. Therefore, the enumeration of matching cuts of $G$ is equivalent to the enumeration of all partitions $\{A,B\}$ of $G$ such that $M=E_G(A,B)$ is a matching cut. Let $T$ be an arbitrary spanning tree of $G$. Then if $M=E_G(A,B)$  is a matching cut of $G$ for a partition $\{A,B\}$, then $E_T(A,B)$ is a matching cut of $T$. Moreover, for two distinct matching cuts $M=E_G(A,B)$ and~$M'=E_G(A',B')$ we have that $E_T(A,B)$ and~$E_T(A',B')$ are distinct as well.  This implies that $\mc(G)\leq \mc(T)\leq F(n+1)-1$ by Lemma~\ref{lem:tree}.

Now we claim that $G$ has $F(n+1)-1$ matching cuts if and only if $G$ is a path. Note that the spanning tree $T$ is arbitrary. If $G$ has a vertex of degree at least three, then $T$ can be chosen in such a way that $T$ is not a path. Then, by Lemma~\ref{lem:tree}, $\mc(G)\leq \mc(T)< F(n+1)-1$.  Assume that the maximum degree of $G$ is at most two. Then $G$ is either a path or a cycle. In the first case, $\mc(G)=F(n+1)-1$ by Observation~\ref{obs:path}. Suppose that $G$ is a cycle $v_0v_1\cdots v_n$ with $v_0=v_n$. Consider the path $P=v_1\cdots v_n$ spanning $G$. Note that every matching cut of $G$ has at least two edges. This implies that there are matching cuts of $P$ that do not correspond to any matching cuts of $G$. In particular, $M=E_P(\{v_1\},\{v_2,\ldots,v_n\})$ is a matching cut of $P$, but  $M'=E_G(\{v_1\},\{v_2,\ldots,v_n\})$ is not a matching cut. This means, that $\mc(G)<\mc(P)=F(n+1)-1$.

To enumerate the matching cuts of $G$, we consider a spanning tree $T$ and enumerate the matching cuts of $T$ using Lemma~\ref{lem:enum-tree}. Then for every matching cut $M=E_T(A,B)$ for a partition $\{A,B\}$ of $V(T)=V(G)$, we verify whether $M'=E_G(A,B)$ is a matching cut of $G$ and output $M'$ is this holds. This means that the matching cuts of a connected graph $G$ can be enumerated in  $\Oh^*(F(n))$ time. This completes the proof for connected graphs.

Assume that $G$ is a disconnected graph with connected components $G_1,\ldots,G_k$, $k\geq 2$, having $n_1,\ldots,n_k$ vertices, respectively.  Observe that $M\subseteq E(G)$ is a matching cut of $G$ if and only if $M=\bigcup_{i=1}^kM_i$, where for every $i\in\{1,\ldots,k\}$, either $M_i$ is a matching cut of $G_i$ or $M_i=\emptyset$. Therefore, using the proved claim for connected graphs, we have that 
\begin{equation}\label{eq:disc-one}
\mc(G)=\prod_{i=1}^k(\mc(G_i)+1)\leq \prod_{i=1}^kF(n_i+1).
\end{equation}
Applying Lemma~\ref{lem:fib} iteratively, we obtain that 
\begin{align}\label{eq:disc-two}
  \prod_{i=1}^kF(n_i+1)\leq& (F(n_1+n_2+1)-1)\prod_{i=3}^kF(n_i)\leq F(n_1+n_2+1)\prod_{i=3}^kF(n_i)-1\leq\cdots\nonumber\\
  \leq&F(n_1+\cdots+n_k+1)-(k-1)=F(n+1)-(k-1).
  \end{align} 
Combining (\ref{eq:disc-one}) and (\ref{eq:disc-two}), we have that 
\begin{equation}\label{eq:disc-three} 
\mc(G)\leq F(n+1)-1.
\end{equation} 

By the proved claim for connected graphs, we have that the inequality (\ref{eq:disc-three}) is strict if one of the connected components is not a path. By inequality (\ref{eq:disc-two}), (\ref{eq:disc-three}) is also strict if $k\geq 3$. If $k=2$ and $n\geq 5$, then either $n_1\geq 3$ and $n_2\geq 3$. By Lemma~\ref{lem:fib}, (\ref{eq:disc-two}) is strict. Hence, if $n\geq 5$, then $\mc(G)< F(n+1)-1$. This implies that if $n\geq 5$, then $\mc(G)=F(n+1)-1$ if and only if $G$ is a path.

Finally, observe that the matching cuts of $G$ can be enumerated by listing the matching cuts of each connected component and combining them (assuming that these lists contain the empty set) to obtain the matching cuts of $G$. Equivalently, we can take the spanning forest $H$ of $G$ obtained by taking the union of spanning trees of $G_1,\ldots,G_k$, respectively. Then we can list the matching cuts of $H$ and output the matching cuts of $G$ corresponding to them. In both cases, the running time is  $\Oh^*(F(n))$. 
\end{proof}

Let us remark that if $n\leq 4$, then besides paths $P_n$, the graphs $K_p+K_q$ for $1\leq p,q\leq 2$ such that $n=p+q$ have $F(n+1)-1$ matching cuts.

Clearly, the upper bound for the maximum number of matching cuts given in Theorem~\ref{thm:general} is an upper bound for the maximum number of minimal and maximal matching cuts. However, the number of minimal or maximal matching cuts may be significantly less than the number of all matching cuts. We conclude this section by stating  the best lower bounds we know for the
maximum number of maximal matching cuts and minimal matching cuts, respectively.

Our lower bound for the maximal matching cuts is achieved for disjoint unions of the cycles on 7 vertices.

\begin{proposition}\label{prop:lower-max}
The graph $G=kC_7$ with $n=7k$ vertices has $14^k=14^{n/7}\geq  1.4579^n$ maximal matching cuts.
\end{proposition}

\begin{proof}
Suppose that $G$ has connected components $G_1,\ldots,G_k$ such that $G_i$ has a matching cut for every $i\in\{1,\ldots,k\}$.
Then  $M\subseteq E(G)$ is a maximal matching cut of $G$ if and only if $M=M_1\cup\cdots\cup M_k$, where $M_i$ is a maximal matching cut of $G_i$ for $i\in\{1,\ldots,k\}$. Observe that $C_7$ has 14 maximal matching cuts formed by matchings with two edges. 
Therefore, $G=kC_7$ has  $14^k$ maximal matching cuts. Since $G$ has $n=7k$ vertices, $14^k=14^{n/7}\geq  1.4579^n$. 
\end{proof}

To achieve a lower bound for the maximum number of minimal matching cuts, we consider the graphs $H_k$ constructed as follows for a positive integer $k$.
\begin{itemize}
\item For every $i\in\{1,\ldots,k\}$, construct two vertices $u_i$ and $v_i$ and a $(u_i,v_i)$-path of length 4.
\item Make the vertices $u_1,\ldots,u_k$ pairwise adjacent, and do the same for $v_1,\ldots,v_k$.
\end{itemize}
\begin{proposition}\label{prop:lower-min}
The number of minimal matching cuts of  $H_k$ with $n=5k$ vertices is at least $4^k=4^{n/5}\geq 1.3195^n$.
\end{proposition}

\begin{proof}
Consider a matching cut $M$ composed by taking one edge of each  $(u_i,v_i)$-path for $i\in\{1,\ldots,k\}$. Clearly, $M$ is a minimal matching cut of $G$. Observe that $H_k$ has 
$4^k$ minimal matching cuts of this form. Since $H_k$ has $n=5k$ vertices, 
$4^k=4^{n/5}\geq 1.3195^n$.
\end{proof}

\section{Enumeration Kernels for the Vertex Cover Number Parameterization}\label{sec:vc}
In this section, we consider the parameterization of the matching cut problems by the vertex cover number of the input graph. Notice that this parameterization is one of the most thoroughly investigated with respect to classical kernelization (see, e.g., the recent paper of Bougeret, Jansen, and Sau~\cite{BougeretJS20} for the currently most general results of this type). However, we are interested in 
enumeration kernels.

Recall that a set of vertices $X\subseteq V(G)$ is a \emph{vertex cover} of $G$ if for every edge $uv\in E(G)$, at least one of its end-vertices is in $X$, that is, $V(G)\setminus X$ is an independent set.  The \emph{vertex cover number} of $G$, denoted by $\vc(G)$, is the minimum size of a vertex cover of $G$. Computing~$\vc(G)$ is \classNP-hard but one can find a 2-approximation by taking the end-vertices of a maximal matching of~$G$~\cite{GareyJ79} (see also~\cite{Karakostas09} for a better approximation) and this suffices for our purposes.  
Throughout this section, we assume that the parameter $k=\vc(G)$ is given together with the input graph. 
Note that for every graph $G$, $\tw(G)\leq \vc(G)$. Therefore, \probMC, \probMinMC, and \probMaxMC can be solved with \classFPT delay when parameterized by the vertex cover number by Proposition~\ref{prop:tw-cw}. 
%


First, we describe the basic kernelization algorithm that is exploited for all the kernels in this subsection.
Let $G$ be a graph that has a vertex cover of size $k$. The case when $G$ has no edges is trivial and will be considered separately. Assume from now that $G$ has at least one edge and $k\geq 1$.

We use the above-mentioned 2-approximation algorithm to find a vertex cover $X$ of size at  most $2k$. Let $I=V(G)\setminus X$. Recall that $I$ is an independent set. 
 Denote by $I_0$, $I_1$, and $I_{\geq 2}$ the subsets of vertices of $I$ of degree 0, 1, and at least 2, respectively. 
We use the following \emph{marking} procedure to label some vertices of $I$.
\begin{itemize}
\item[(i)] \emph{Mark} an arbitrary vertex of $I_0$ (if it exists).
\item[(ii)] For every $x\in X$, \emph{mark} an arbitrary vertex of $N_G(x)\cap I_1$ (if it exists). 
\item[(iii)] For every two distinct vertices $x,y\in X$, select an arbitrary set of $\min\{3,|(N_G(x)\cap N_G(y))\cap I_{\geq 2}|\}$ vertices in $I_{\geq 2}$ that are adjacent  to both $x$ and $y$, and \emph{mark} them for the pair $\{x,y\}$. 
\end{itemize}
Note that a vertex of $I_{\geq 2}$ can be marked for distinct pairs of vertices of $X$.
Denote by $Z$ the set of marked vertices of $I$. Clearly,  $|Z|\leq 1+|X|+3\binom{|X|}{2}$.   
We define $H=G[X\cup Z]$. Notice that $|V(H)|\leq |X|+|Z|\leq 1+2|X|+3\binom{|X|}{2}\leq 6k^2+k+1$. This completes the description of the basic kernelization algorithm that returns $H$. It is straightforward to see that $H$ can be constructed in polynomial time.

It is easy to see that $H$ does not keep the information about all matching cuts in $G$ due to the deleted vertices. However, the crucial property is that $H$ keeps all matching cuts of $G'=G-(I_0\cup I_1)$.  Formally, we define $H'=H-(I_0\cup I_1)$ and show the following lemma.

\begin{lemma}\label{lem:kernel-vc}
A set of edges $M\subseteq E(G')$ is a matching cut of $G'$ if and only if $M\subseteq E(H')$ and $M$ is a matching cut of $H'$.
\end{lemma}  

\begin{proof}
Suppose that $M\subseteq E(G')$ is a matching cut of $G'$ and assume that $M=E_{G'}(A,B)$ for a partition $\{A,B\}$ of $V(G')$. We show that $M\subseteq E(H')$. 
For the sake of contradiction, suppose that there is some edge $uv\in M\setminus E(H')$. This means that either $u\notin V(H')$ or $v\notin V(H')$. By symmetry, we can assume without loss of generality that $u\notin V(H')$. Then $u\in I_{\geq 2}\setminus Z$, that is, $u$ is an unmarked vertex of $I_{\geq 2}$. Recall that every vertex of $I_{\geq 2}$ has degree at least two. Hence, $u$ has a neighbor $w\neq v$. Because $M$ is a matching cut, $w\in A$. Notice that $w,v\in X$. Because $u$ is unmarked and adjacent to both $w$ and $v$, there are three vertices $z_1,z_2,z_3\in Z$ that are marked for the pair $\{w,v\}$. Since either $A$ or $B$ contains at least two of the vertices $z_1,z_2,z_3$, either $w$ or $v$ has at least two neighbors in $B$ or $A$, respectively. This contradicts that $M$ is a matching cut and we conclude that $M\subseteq E(H')$. Since $H'$ is an induced subgraph of $G'$, $M$ is a matching cut of $H'$.

For the opposite direction, assume that $M$ is a matching cut of $H'$.  Let $M=E_{H'}(A,B)$ for a partition $\{A,B\}$ of $V(H')$.   We claim that for every $v\in V(G')\setminus V(H')$, either $v$ has all its neighbors in $A$ or $v$ has all its neighbors in $B$. The proof is by contradiction. Assume that $v\in V(G')\setminus V(H')$ has a neighbor $u\in A$ and a neighbor $w\in B$. Then $v$ is an unmarked vertex of $I_{\geq 2}$, and  $u,w\in X$. We have that for the pair $\{u,w\}$, there are three marked vertices $z_1,z_2,z_3\in I_{\geq 2}$ that are adjacent to both $u$ and $w$. Since $z_1,z_2,z_3$ are marked, $z_1,z_2,z_3\in V(H')$. In the same way as above,  either $A$ or $B$ contains at least two of the vertices $z_1,z_2,z_3$ and this implies that either $u$ or $v$ has at least two neighbors in the opposite set of the partition. This contradicts the assumption that $M$ is a matching cut of $H'$. Since for every $v\in V(G')\setminus V(H')$, either $v$ has all its neighbors in $A$ or $v$ has all its neighbors in $B$, $M$ is a cut of $G'$, that is, $M$ is a matching cut of $G'$.
\end{proof}

To see the relations between matching cuts of $G$ and $H$, we define a special equivalence relation for the subsets of edges of $G$. For a vertex $x\in X$, let $L_x=\{xy\in E(G)\mid y\in I_1 \}$, that is, $L_x$ is the set of pendant edges of $G$ with exactly one end-vertex in the vertex cover.  Observe that if $L_x\neq \emptyset$, then there is $\ell_x\in L_x$ such that $\ell_x\in E(H)$, because for every $x\in X$, a neighbor in $I_1$ is marked if it exists. We define $L=\bigcup_{x\in X}L_x$. 
Notice that each matching cut of $G$ contains at most one edge of every $L_x$.
We say that two sets of edges $M_1$ and $M_2$ are \emph{equivalent} if $M_1\setminus L=M_2\setminus L$ and 
for every $x\in X$, $|M_1\cap L_x|=|M_2\cap L_x|$.  It is straightforward to verify that the introduced relation is indeed an equivalence relation.
It is also easy to see that if $M$ is a matching cut of $G$, then every $M'\subseteq E(G)$ equivalent to $M$ is a matching cut.
 We show the following lemma.

\begin{lemma}\label{lem:equiv}
A set of edges $M\subseteq E(G)$ is a  matching cut (minimal or maximal matching cut, respectively)  of $G$ if and only if $H$ has a matching cut (minimal or maximal matching cut, respectively) $M'$   equivalent to $M$. 
\end{lemma}

\begin{proof}
We prove the lemma for matching cuts.  

For the forward direction, let $M$ be a matching  cut of $G$. We show that there is a matching cut $M'$ of $H$ that is a matching cut of $G$ equivalent to $M$. 
If $M=\emptyset$, then $G$ is disconnected. Notice that, by the construction, $H$ is disconnected as well. Hence, $M'=M$ is a matching cut of $H$. Clearly, $M'$ is equivalent to $M$.
Assume that $M\neq\emptyset$.
We construct $M'$ from $M$ by the following operation: for every $x\in X$ such that $M\cap L_x\neq\emptyset$, replace the unique edge of $M$ in $L_x$ by $\ell_x$. By the definition, $M'$ is a matching cut that is equivalent to $M$. We show that $M'$ is a matching cut of $H$. Let $M_1=L\cap M'$ and $M_2=M'\setminus M_1$. We have that either $M_2=\emptyset$ or $M_2$ is a nonempty matching cut of $G'$.  If $M_2=\emptyset$, then it  is straightforward to see that $M'=M_1$ is a matching cut of $H$. 
If $M_2\neq\emptyset$, then by Lemma~\ref{lem:kernel-vc}, $M_2$ is a matching cut of $H'$. This implies that $M'$ is a matching cut of $H$.

For the the opposite direction, assume that $M'$ is a matching cut of $H$. If $M'=\emptyset$, then $H$ is disconnected. By construction, $G$ is disconnected as well and $M'$ is a matching cut of $G$. Let $M'\neq\emptyset$. Let also $M_1=M'\cap L$ and $M_2=M'\setminus M_1$. If $M_2=\emptyset$, then $M'=M_1$ is a matching cut of $G$. If $M_2\neq\emptyset$, then $M_2$ is a matching cut of $H'$. By Lemma~\ref{lem:kernel-vc}, $M_2$ is a matching cut of $G'$. Then $M'$ is a matching cut of $G$. It remains to notice that for every $M\subseteq E(G)$ equivalent to $M'$, $M$ is a matching cut of $G$. 

For minimal and maximal matching cuts, the arguments are the same. It is sufficient to note that if $M$ and $\hat{M}$ are matching cuts of $G$ such that $M\subset\hat{M}$, then their equivalent matching cuts $M'$ and $\hat{M}'$ of $H$ constructed in the proof for the forward direction satisfy the same inclusion property $M'\subset \hat{M}'$. 
\end{proof}

We use Lemma~\ref{lem:equiv} to obtain our kernelization results. For \probMinMC, we show that the problem admits a fully-polynomial enumeration kernel, and we prove that   \probMaxMC and \probMC have 
polynomial-delay enumeration kernels.

\begin{theorem}\label{thm:vc-kern}
\probMinMC admits a \fp{} enumeration kernel and \probMC and \probMaxMC admit 
 \pd{} enumeration kernels with $\Oh(k^2)$
vertices when parameterized by the vertex cover number $k$ of the input graph. 
\end{theorem}

\begin{proof}
Let $G$ be a graph with $\vc(G)=k$. If $G=K_1$, then the kernelization algorithm returns $H=G_1$ and the solution-lifting algorithm is trivial as $G$ has no matching cuts. Assume that~$G$ has at least $2$~vertices. If $G$ has no edges, then the empty set is the unique matching cut of $G$. Then the kernelization algorithm returns $H=2K_1$, and the solution-lifting algorithm outputs the empty set for the empty matching cut of $H$. Thus, we can assume without loss of generality that $G$ has at least one edge and $k\geq 1$.

We use the same basic kernelization algorithm that constructs $H$ as described above and output $H$ for all the problems. 
Recall that $|V(H)|\leq 6k^2+k+1$. 
The kernels differ only in the solution-lifting algorithms. These algorithms exploit  Lemma~\ref{lem:equiv} and for every matching cut (minimal or maximal matching cut, respectively) $M$ of $H$, they list the equivalent matching cuts of $G$. Lemma~\ref{lem:equiv} guarantees that the families of matching cuts (minimal or maximal matching cuts, respectively) constructed for every matching cut of $H$ compose the partition of the sets of  matching cuts (minimal or maximal matching cuts, respectively) of $G$.  This is exactly the property that  is required by the definition of  
a \fp{} (\pd{})
enumeration kernel. To describe the algorithm, we use the notation defined in this section. 
 
First, we consider \probMinMC. Let $M$ be a minimal matching cut of $H$. If $M\cap L=\emptyset$, then $M$ is the unique matching cut of $G$ that is equivalent to $M$, and our algorithm outputs~$M$. Suppose that $M\cap L\neq\emptyset$. Then by the minimality of $M$, $M=\{\ell_x\}$ for some $x\in X$, because every edge of $L$ is a matching cut. Then
the sets $\{e\}$ for every $e\in L_x$ are exactly the matching cuts equivalent to~$M$. Clearly, we have at most $n$ such matching cuts and they can be listed in linear time. This implies that condition (ii) of the definition of a \fp{} enumeration kernel is fulfilled. Thus, \probMinMC has a \fp{} enumeration kernel with at most $6k^2+k+1$  vertices.

Next, we consider \probMaxMC and \probMC. The solution-lifting algorithms for these problems are the same. Let $M$ be a (maximal) matching cut of $H$. Let also $M_1=M\cap L$ and $M_2=M\setminus M_1$. If $M_1=\emptyset$, then  $M$ is the unique matching cut of $G$ that is equivalent to $M$, and our algorithm outputs $M$. Assume from now that $M_1\neq \emptyset$. Then there is $Y\subseteq X$ such that $M_1=\{\ell_x\mid x\in Y\}$.  
We use the recursive algorithm 
\textsc{Enum Equivalent} (see Algorithm~\ref{alg:max}) that takes as an input a matching $S$ of $G$ and $W\subseteq Y$ and outputs the equivalent matching cuts $M'$ of $G$ such that 
(i) $S\subseteq M'$, (ii) $M'$ is equivalent to $M$,  and (iii) the constructed matchings $M'$  differ only by some edges of the sets $L_x$ for $x\in W$. Initially, $S=M_2$ and $W=Y$. 

\begin{algorithm}[ht]
\caption{$\textsc{Enum Equivalent}(S,W)$}\label{alg:max}
\If{$W=\emptyset$}{\Output{$S$}}
\ElseIf{$S\neq\emptyset$}
{
select arbitrary $x\in W$\;
\ForEach{$e\in L_x$}
{$\textsc{Enum Equivalent}(S\cup\{e\}, W\setminus\{x\})$
}
}
\end{algorithm}

To enumerate the matching cuts equivalent to $M$, we call $\textsc{Enum Equivalent}(M_2,Y)$. We claim that 
 $\textsc{Enum Equivalent}(M_2,Y)$ enumerates the matching cuts of $G$ that are equivalent to $M$ with $\Oh(n)$ delay. 

By the definition of the equivalence and Lemma~\ref{lem:equiv}, every matching cut $M'$ of $G$ that is equivalent to $M$ can be written as $M'=M_2\cup\{e_x\mid x\in Y\}$, where $e_x$ is an edge of $L_x$ for $x\in Y$.  Then to see the correctness of \textsc{Enum Equivalent}, observe the following.
If $W\neq\emptyset$, then the algorithm picks a vertex $x\in W$. Then for every edge $e\in L_x$, it enumerates the matching cuts containing $S$ and $e$. This means that our algorithm is, in fact, a standard \emph{backtracking} enumeration algorithm (see~\cite{Marino15}) and immediately implies that the algorithm lists all the required matching cuts exactly once. Since the depth of the recursion is at most~$n$ and the algorithm always outputs a matching cut for each leaf of the search tree, the delay is~$\Oh(n)$.
This completes the proof of the \pd{} enumeration kernel for \probMaxMC and \probMC.

To conclude the proof of the theorem, let us remark that, formally, the solution-lifting algorithms described in the proof require $X$. However, in fact, we use only sets $L_x$ that can be computed in polynomial time for given $G$ and $H$.  
\end{proof}

Notice that Theorem~\ref{thm:vc-kern} is tight in the sense that 
\probMaxMC and \probMC do not admit \fp{} enumeration kernels for the parameterization by the vertex cover number. To see this, let $k$ be a positive integer and consider the $n$-vertex graph $G$, where $n>k$ is divisible by~$k$, that is the union of $k$ stars $K_{1,p}$ for $p=n/k-1$.  Clearly, $\vc(G)=k$. We observe that $G$ has $p^k=(n/k-1)^k$ maximal matching cut that are formed by picking one edge from each of the $k$ stars. Similarly, $G$ has $(p+1)^k=(n/k)^k$ matching cuts obtained by picking at most one edge from each star. In both cases,  this means that the (maximal) matching cuts cannot be enumerated by an \classFPT algorithm. By Theorem~\ref{obs:natural}, this rules out the existence of a \fp{} enumeration kernel.

By Theorems~\ref{thm:vc-kern} and~\ref{obs:natural}, we have that the minimal matching cuts of a graph $G$ can be enumerated in $2^{\Oh(\vc(G)^2)}\cdot n^{\Oh(1)}$ time by the applying the enumeration algorithm from Theorem~\ref{thm:general} for  $H$. Similarly, the (maximal) matching cuts can be listed with $2^{\Oh(\vc(G)^2)}\cdot n^{\Oh(1)}$ delay. We show that this running time can be improved and the dependence on the vertex cover number can be made single exponential.

\begin{theorem}\label{thm:upper-vc}
The minimal matching cuts of an $n$-vertex graph $G$ can be enumerated in $2^{\Oh(\vc(G))}\cdot n^{\Oh(1)}$ time, and the (maximal) matching cuts of $G$ can be enumerated with 
$2^{\Oh(\vc(G))}\cdot n^{\Oh(1)}$ delay.
\end{theorem}  
  
\begin{proof} 
Recall that our kernelization algorithm that is the same for all the problems, given a graph $G$ with $\vc(G)=k$, constructs the graph~$H$ which gives 
a \fp{} (\pd{})
enumeration kernel, by Theorem~\ref{thm:vc-kern}. By Definition~\ref{def:enum-kernel},  it is thus sufficient to show that~$H$ has $2^{\Oh(k)}$ matching cuts that can be listed in  $2^{\Oh(k)}$ time.  We do it using the structure of $H$ following the notation introduced in the beginning of the section. 

Notice that every  matching cut of $H$ can be written as $M=M_1\cup M_2$, where $M_1=\{\ell_x\mid x\in Y\}$ for some $Y\subseteq X$ and $M_2$ is either empty or a nonempty matching cut of $H'=H-(I_0\cup I_1)$. Observe that because $|X|\leq 2k$, $H$ has at most $2^{2k}$ sets of edges of the form $\{\ell_x\mid x\in Y\}$ for some $Y\subseteq X$ and all such sets can be listed in $2^{\Oh(k)}$ time. Hence, it is sufficient to show that $H'$ has $2^{\Oh(k)}$ matching cuts that can be enumerated in  $2^{\Oh(k)}$ time.

Let $M$ be a nonempty matching cut of $H'$ and let $\{U,W\}$ be a partition of $V(H')$ such that $M=E(U,W)$. Recall that each vertex $v\in V(H')$ that is not a vertex of $X$ belongs to $I_{\geq 2}$, that is, has at least two neighbors in $X$. Therefore, $U\cap X$ and $W\cap X$ are nonempty.   
Since the empty matching cut can be listed separately  if it exists, it is sufficient to enumerate matching cuts of the form $M=E(U,W)$ with nonempty $U\cap X$ and $W\cap X$. For this we consider all partitions $\{A,B\}$ of $X$, and for each partition, enumerate matching cuts $M=E(U,W)$, where $A\subseteq U$ and $B\subseteq W$. Since $|X|\leq 2k$, there are at most $2^{2k}$ partitions $\{A,B\}$ of $X$ and they can be listed in $2^{\Oh(k)}\cdot n^{\Oh(1)}$ time. 

Assume from now that a partition $\{A,B\}$ of $X$ is given. We construct a recursive branching algorithm \textsc{Enumerate MC}$(A',B')$, whose  input consists of two disjoint sets $A'$ and $B'$ such that $A\subseteq A'$ and $B\subseteq B'$, and the algorithm outputs all  matching cuts of the form $M=E(U,W)$ with $A\subseteq U$ and $B\subseteq W$.  It is convenient for us to write down the algorithm as a series of steps and \emph{reduction} and \emph{branching} rules as it is  common for exact algorithms (see, e.g.,~\cite{FominK10}). The algorithm maintains the set $S$ of the end-vertices of the edges of $E(A',B')$ in $X$, and we implicitly assume that $S$ is recomputed at each step if necessary. We say that $S$ is the set of \emph{saturated} vertices of $X$. Initially, $S$ is the set of end-vertices of $E(A,B)$.

Clearly, if $M=E(A',B')$ is not a matching, then $M$ cannot be extended to a matching cut and we can stop considering $\{A',B'\}$.

\begin{step}\label{st:one}
If $E(A',B')$ is not a matching, then quit. 
\end{step}

If $\{A',B'\}$ is a partition of $V(H')$, then the algorithm finishes its work. Note that since the algorithm did not quit in the previous step, $E(A',B')$ is a matching.

\begin{step}
If $\{A',B'\}$ is a partition of $V(H')$ and $M=E(A',B')$ is a matching cut, then output $M$ and quit.
\end{step}

From now, we can assume that there is $v\in V(H')\setminus (A'\cup B')$. Recall that by the construction of $H'$, $v\in I_{\geq 2}$, that is, $v$ has at least two neighbors in $X$. 

Clearly, if $v\in V(H')\setminus (A'\cup B')$ has at least two neighbors in $A'$, then $v\in U$ for every matching cut $E(U,W)$ with $A'\subseteq U$. This gives us the following reduction rule.

\begin{reduction}
If there is $v\in V(H')\setminus (A'\cup B')$ such that $|N_{H'}(v)\cap A'|\geq 2$ ($|N_{H'}(v)\cap B'|\geq 2$, respectively), then call $\textsc{Enumerate MC}(A'\cup\{v\},B')$ 
($\textsc{Enumerate MC}(A',B'\cup\{v\})$, respectively).
\end{reduction}

If $v$ has at least two neighbors in both $A'$ and $B'$, we place $v$ in $A'$; note that we stop in Step~\ref{st:one} afterwards.
From now, we assume that the rule cannot be applied, that is, every vertex $v\in V(H')\setminus (A'\cup B')$ has exactly one neighbor $x$ in $A'$ and exactly one neighbor $y$ in $B'$. Notice that either $vx$ or $vy$ should be in a (potential) matching cut. This gives the following rules.

\begin{reduction}
If there is $v\in V(H')\setminus (A'\cup B')$ with neighbors $x\in A'$ and $y\in B'$ such that $x\in S$ ($y\in S$), then call  $\textsc{Enumerate MC}(A'\cup\{v\},B')$ 
($\textsc{Enumerate MC}(A',B'\cup\{v\})$, respectively).
\end{reduction} 

If both neighbors of $v$ are saturated, we place $v$ in $A'$; notice that we stop in Step~\ref{st:one} in  $\textsc{Enumerate MC}(A'\cup\{v\},B')$. 
The rule is safe, because every saturated vertex of $X$ can be adjacent to only one edge of a matching cut. From now, we can assume that the neighbors of~$v$ are not saturated. In this case, we branch on two possibilities for $v$.

\begin{branch}\label{br:one}
If there is $v\in V(H')\setminus (A'\cup B')$ with neighbors $x\in A'$ and $y\in B'$ such that $x,y\notin S$, then call  
\begin{itemize}
\item $\textsc{Enumerate MC}(A'\cup\{v\},B')$, 
\item $\textsc{Enumerate MC}(A',B'\cup\{v\})$.
\end{itemize}
\end{branch}
This finishes the description of the algorithm; its correctness  follows directly from the discussion of the reduction and branching rules. To evaluate the running time, note that for every recursive call of \textsc{Enumerate MC} in Branching Rule~\ref{br:one}, we increase $S$ by including either $x$ or $y$ into the set of saturated vertices of $X$. Since $|S|\leq |X|\leq 2k$, the depth of the search tree is upper bounded by $2k$. Because we have two branches in Branching Rule~\ref{br:one}, the number of leaves in the search tree is at most $2^{2k}$. Observing that all the steps and rules can be executed in polynomial time, we obtain that the total running time is $2^{\Oh(k)}$ using the standard analysis of the running time of recursive branching algorithm (see~\cite{FominK10}).  

Since the search tree has at most $2^{2k}$ leaves, the number of matching cuts produced by the algorithm from the given partition $\{A,B\}$ of $X$ is at most $2^{2k}$.  Because the number of partitions is at most $2^{2k}$, we have that $H'$ has at most $2^{\Oh(k)}$ matching cuts. Then the number of matching cuts of $H$ is $2^{\Oh(k)}$.
Combining \textsc{Enumerate MC} with the previous steps for the enumeration of the matching cuts, we conclude that the matching cuts  of $H$ can be listed in 
 $2^{\Oh(k)}$ time.
\end{proof}

 We complement Theorem~\ref{thm:upper-vc} by the following lower bound that shows that exponential in $\vc(G)$ time for \probMinMC is unavoidable.
 
 \begin{proposition}\label{prop:lower-vc}
There is an infinite family of graphs whose number of minimal matching cuts is~$\Omega(2^{\vc(G)})$.
\end{proposition}

\begin{proof}
  Consider a graph consisting of $k$ disjoint copies of $P_3$s ~$u_iv_iw_i$,~$1\le i \le k$, and two
  additional vertices~$s$ and~$t$ such that~$s$ is adjacent to each~$u_i$,~$1\le i\le k$,
  and~$t$ is adjacent to each~$w_i$,~$1\le i\le k$. We call~$s u_i v_i w_i
  t$,~$1\le i \le k$, a 5-path of this graph. The vertex cover number of this graph
  is~$k+2$. Now consider the matching cuts which contain at most one edge of each
  5-path. Such a matching cut~$M$ contains also at least one edge of each 5-path:
  otherwise,~$s$ and~$t$ are connected via some 5-path and, since all other vertices are
  connected to~$s$ or~$t$, the graph is connected after the removal of~$M$.

  Thus, each matching cut~$M$ that contains exactly one edge of each 5-path is
  minimal. Now consider any index set~$I\subseteq \{1,\ldots ,k\}$ and observe
  that~$M_I=\{u_iv_i\mid i\in I\}\cup \{v_iw_i\mid i\in \{1,\ldots,k\} \setminus
  I\}$ is a minimal matching cut. Since there are~$2^k$ possibilities for~$I$, the number
  of minimal matching cuts is~$\Omega(2^{k})=\Omega(2^{\vc(G)})$.
\end{proof}

We conclude this section by showing that Theorem~\ref{thm:vc-kern} can be generalized to the weaker parameterization by the \emph{twin-cover} number, introduced by Ganian~\cite{Ganian11,Ganian15} as a generalization of a vertex cover. Recall that two vertices $u$ and $v$ of a graph $G$ are \emph{true twins} if $N[u]=N[v]$. A set of vertices $X$ of a graph $G$ is said to be a \emph{twin-cover} of $G$ if for every edge $uv$ of $G$, at least one of the following holds: (i) $u\in X$ or $v\in X$ or (ii) $u$ and $v$ are true twins.  The \emph{twin-cover} number, denoted by $\tc(G)$, is the minimum size of a twin-cover. Notice that  $\tc(G)\leq \vc(G)$ and $\tc(G)\geq \cw(G)+2$ for every $G$~\cite{Ganian11, Ganian15}. 
This means that the parameterization by the twin-cover number is weaker than the parameterization by the vertex cover number but stronger than the cliquewidth parameterization.
It is convenient for us to consider a parameterization that is weaker than the twin-cover number. 
Let $\mathcal{X}=\{X_1,\ldots,X_r\}$ be the partition of $V(G)$ into the classes of true twins. Note that $\mathcal{X}$ can be computed in linear time using an algorithm for computing a modular decomposition~\cite{TedderCHP08}. Then we can define the \emph{true-twin quotient} graph $\mathcal{G}$ with respect to $\mathcal{X}$, that is, the graph with the node set $\mathcal{X}$ such that two classes of true twins $X_i$~and~$X_j$ are adjacent in $\mathcal{G}$ if and only if the vertices of $X_i$ are adjacent to the vertices of $X_j$ in~$G$. 
Then it can be seen that $\tc(G)\geq\vc(\mathcal{G})$. We prove that \probMinMC admits a \fp{}  enumeration kernel and \probMaxMC and \probMC admit \pd{} enumeration kernels for the parameterization by the vertex cover number of the true-twin quotient graph of the input graph.  In particular, this
implies the same kernels for the twin-cover parameterization.

As the first step, we show the following corollary of Theorem~\ref{thm:vc-kern}.

\begin{corollary}\label{cor:vc-kern-A}
Let $\mathcal{C}$ be the class of graphs $G=\hat{G}+sK_2$, where $\vc(\hat{G})\leq k$. 
Then  \probMinMC admits a \fp{} enumeration kernel and \probMC and \probMaxMC admit  
\pd{} enumeration kernels with $\Oh(k^2)$ 
vertices for graph in $\mathcal{C}$. 
\end{corollary} 

\begin{proof}
Let $G=\hat{G}+sK_2$ and $\vc(G)\leq k$. The claim is an immediate corollary of Theorem~\ref{thm:vc-kern} if $s=0$, that is, if $\vc(G)\leq k$. Assume that $s\geq 1$.

For \probMinMC, it is sufficient to observe that $G$ is disconnected if $s\geq 1$ and, therefore, the empty set is the unique minimal matching cut. Then the kernelization algorithm outputs $2K_1$ and the solution-lifting algorithm outputs the empty set.

For \probMC and \probMaxMC, let $G'=\hat{G}+K_2$. Denote by $e$ the unique edge of the copy of $K_2$. Observe that $\vc(G')=\vc(\hat{G})+1\leq k+1$. We apply Theorem~\ref{thm:vc-kern} for $G'$ and obtain 
a \pd{} enumeration kernel $H$ with $\Oh(k^2)$ vertices. It is easy to observe that every maximal matching cut of $G'$ contains $e$ and every maximal matching cut of $G$ contains all the edges of the $s$ copies of $K_2$. Then for \probMaxMC, we modify the solution-lifting algorithm as follows: for every maximal matching cut that the algorithm outputs for a maximal matching cut of $H$ and the graph $G'$, 
 we construct the matching cut of $G$ by replacing $e$ in the matching cut by $s$ edges of the $s$ copies of $K_2$ in $G$. 
 For \probMC, the modification of the solution-lifting algorithm is a bit more complicated. Let $M$ be a matching cut produced    
by the solution-lifting algorithm for a matching cut of $H$ and the graph $G'$. If $e\notin M$, then the modified solution-lifting algorithm just outputs $M$. Otherwise, if $e\in M$, the solution-lifting algorithm outputs the matching cuts $(M\setminus\{e\})\cup L$, where $L$ 
 is a nonempty subset of edges of the $s$ copies of $K_2$ in $G$. Since all nonempty subsets of a finite set can be enumerated with polynomial delay by very basic combinatorial tools (see, e.g.,~\cite{Marino15}), the obtained modification of the solution-lifting algorithm is a solution-lifting algorithm for $H$ and $G$.
\end{proof}
 
 Using Corollary~\ref{cor:vc-kern-A}, we prove the following theorem.
\begin{theorem}\label{thm:tc-kern}
\probMinMC admits a \fp{} enumeration kernel and \probMC and \probMaxMC admit  
\pd{} enumeration kernels with $\Oh(k^2)$
vertices when parameterized 
by the vertex cover number of the true-twin quotient graph of the input graph.  
\end{theorem}

\begin{proof}
Let $G$ be a graph and let $\mathcal{G}$ be its true-twin quotient graph with
$\vc(\mathcal{G})=k$. Let also  $\mathcal{X}=\{X_1,\ldots,X_r\}$ be the partition of $V(G)$ into the classes of true twins (initially, $r=k$). 
Recall that $\mathcal{X}$ can be computed in linear time~\cite{TedderCHP08}. 

We apply the following series of reduction rules. All these rules use the property that if $K$ is a clique of $G$ of size at least three, then either $K\subseteq A$ or $K\subseteq B$ for every partition $\{A,B\}$ of $V(G)$ such that $M=E(A,B)$ is a matching cut.

\begin{reduction}\label{red:clique-red}
If there is $i\in\{1,\ldots,r\}$ such that  $|X_i|\geq 4$, then delete $|X_i|-3$ vertices of $X_i$. 
\end{reduction} 
 
To see that the rule is safe, let $X_i'$ be the clique obtained from $X_i$ by the rule for some $i\in\{1,\ldots,r\}$, and denote by $G'$ the obtained graph.  We claim that $M$ is a matching cut of $G$ if and only if $M$ is a matching cut of $G'$. Let $M=E_g(A,B)$, where $\{A,B\}$  is a partition of~$V(G)$. We have that either $X_i\subseteq A$ or $X_i\subseteq B$. By symmetry, we can assume without loss of generality that $X_i\subseteq A$. Note that since $|X_i|\geq 2$, the vertices of  $N_G(X_i)$ are in $A$. Otherwise, we would have a vertex $v\in B$ with at least two neighbors in $A$. This implies that no edge of $M$ is incident to a vertex of $X_i$ and, therefore, $M\subseteq E(G')$. Since $G'$ is an induced subgraph of $G$, $M$ is a matching cut of $G'$. For the opposite direction, the arguments are essentially the same. If  $\{A',B'\}$  is a partition of $V(G')$ with $M=E_{G'}(A',B')$, then we can assume without loss of generality that $X_i'\subseteq A'$. Then $N_{G'}(X_i)\subseteq A'$. This implies that for  $A=A'\cup X_i$ and $B=B'$, $E_G(A,B)=E_{G'}(A',B')=M$, that is, $M$ is a matching cut of $G$. Summarizing, we conclude that the enumeration of matching cuts (minimal or maximal matching cuts, respectively) for $G$ is equivalent to their enumeration for $G'$. This means that Reduction Rule~\ref{red:clique-red} is safe. 

We apply   Reduction Rule~\ref{red:clique-red} for all classes of true twins of size at least four. To simplify notation, we use $G$ to denote the obtained graph and $X_1,\ldots,X_r$ is used to denote the obtained classes of true twins. We have that $|X_i|\leq 3$ for $i\in\{1,\ldots,r\}$. We show that we can reduce the size of some classes even further. This is straightforward for classes $X_i$ of size three that induce connected components of $G$.

\begin{reduction}\label{red:comp-red}
If there is $i\in\{1,\ldots,r\}$ such that  $|X_i|=3$ and $G[X_i]$ is a connected component of $G$, then delete arbitrary two vertices of $X_i$. 
\end{reduction}    
  
Further, we delete some classes having a unique neighbor. 
  
\begin{reduction}\label{red:comp-del}
If there is $i\in\{1,\ldots,r\}$ such that $|X_i|\geq 2$ and $|N_G(X_i)|=1$, then delete the vertices of $X_i$. 
\end{reduction}   
  
To prove safeness, assume that the rule is applied for $X_i$. Let $G'$ be the graph obtained by the deletion of $X_i$ and let $y$ be the unique vertex of $N_G(X_i)$. We show that $M$ is a matching cut of $G$ if and only if $M$ is a matching cut of $G'$. Assume first that   
  $M=E_G(A,B)$ is a matching cut of $G$ for a partition $\{A,B\}$ of $V(G)$.  Since $X_i\cup \{y\}$ is a clique of size at least three, either $X_i\cup\{y\}\subseteq A$ or $X_i\cup\{y\}\subseteq B$. By symmetry, we can assume without loss of generality that   $X_i\cup\{y\}\subseteq A$. Then no edge of $M$ is incident to a vertex of $X_i$. This implies that $M\subseteq E(G')$. Since $G'$ is an induced subgraph of $G$,  $M$ is a matching cut of $G'$. For the opposite direction, assume that $M=E_{G'}(A',B')$ is a matching cut of $G'$ for a partition $\{A',B'\}$ of $V(G')$. We can assume without loss of generality that $y\in A'$. Then it is straightforward  
 to see that $M=E_G(A'\cup X_i,B')$, that is, $M$ is a matching cut of $G$. We obtain that Reduction Rule~\ref{red:comp-del} is safe, because the enumeration of   matching cuts (minimal or maximal matching cuts, respectively) for $G$ is equivalent to their enumeration for $G'$.  
  
We apply Reduction Rules~\ref{red:comp-red} and \ref{red:comp-del} for all classes $X_i$ satisfying their conditions.
We use the same convention as before, and use $G$ to denote the obtained graph and $X_1,\ldots,X_r$ is used to denote the obtained sets of true twins.   
  
 Finally, we reduce the size of some classes that have at least two neighbors. Recall that $\tc(G)=k$. This means that $\vc(\mathcal{G})=k$ for the quotient graph $\mathcal{G}$ constructed for $\mathcal{X}=\{X_1,\ldots,X_r\}$. We compute $\mathcal{G}$ and 
 use, say a 2-approximation algorithm~\cite{GareyJ79}, to find a vertex cover $\mathcal{Z}$ of size at  most $2k$. Let $\mathcal{I}=V(\mathcal{G})\setminus \mathcal{Z}$. Recall that $\mathcal{I}$ is an independent set.

 \begin{reduction}\label{red:comp-two}
If there is $i\in\{1,\ldots,r\}$ such that $X_i\in\mathcal{I}$,  $|X_i|\geq 2$, and $|N_G(X_i)|\geq 2$, then delete arbitrary $|X_i|-1$ vertices of $X_i$ and make the vertices of $N_G(X_i)$ pairwise adjacent by adding edges. 
\end{reduction}  
 
To show that the rule is safe assume that $X_i\in\mathcal{I}$,  $|X_i|\geq 2$, and $|N_G(X_i)|\neq 2$ for some $i\in\{1,\ldots,r\}$ and the rule is applied for $X_i$. Denote by $G'$ the graph obtained by the application of the rule, and let $x\in X_i$ be the vertex of $G'$. 
We claim that  $M$ is a matching cut of $G$ if and only if $M$ is a matching cut of $G'$.

For the forward direction, let $M=E_G(A,B)$ be a matching cut of $G$ for a partition $\{A,B\}$ of $V(G)$. For every $y\in N_G(X_i)$, $Z_y=X_i\cup\{y\}$ is a clique of size at least three in $G$. Therefore, either $Z_y\subseteq A$ or $Z_y\subseteq B$. By symmetry, we can assume without loss of generality that $Z_y\subseteq A$ for all $y\in N_G(X_i)$, that is, $N_G[X_i]\subseteq A$ and, moreover, the edges of $M$ are not incident to the vertices of $X_i$. Therefore, $M\subseteq E(G')$ and the edges between the vertices of $N_{G}[X_i]$ that may be added by the rule have their end-vertices in $A$. This implies that $M$ is a matching cut of $G'$.

Assume now that $M=E_{G'}(A',B')$ is a matching cut of $G'$ for  a  partition $\{A',B'\}$ of $V(G')$. Assume without generality that $x\in A$. Let also $K=N_{G'}[x]$; note that $K$ is a clique of $G'$. Since $|N_G(X_i)|\geq 2$, $|K|\geq 3$. Hence, $K\subseteq A$. Notice also that the edges of $M$ are not incident to $x$ and are not edges of $G'[K]$. Because $N_G(z)\cap V(G')=N_G(x)\cap V(G')=N_{G'}(x)$ for every $z\in X_i$, we obtain that $M=E_G(A'\cup X_i,B')$ and $M$ is a matching cut of $G$. We conclude that   the enumeration of   matching cuts (minimal or maximal matching cuts, respectively) for $G$ is equivalent to their enumeration for $G'$.  Therefore, Reduction Rule~\ref{red:comp-two} is safe.

We apply Reduction Rule~\ref{red:comp-two} for the classes in $\mathcal{I}$ exhaustively. Denote by $G^*$ the obtained graph and let $X_1^*,\ldots,X_r^*$ be the constructed classes of true twins. We use $\mathcal{I}^*$ to denote the family of sets of true twins obtained from the sets of $\mathcal{I}$. Note that the sets of $\mathcal{Z}$ are not modified by Reduction Rule~\ref{red:comp-two}.

We show the following claim summarizing the properties of the obtained sets of true twins.

\begin{claim}\label{cl:final-red}
For every $X_i^*\in \mathcal{I}^*$, either $G[X_i^*]$ is a connected component of $G^*$ and $|X_i^*|=2$ or $|X_i|=1$, and for every $X_i^*\in\mathcal{Z}$, $|X_i^*|\leq 3$.
\end{claim}
 
\begin{proof}[Proof of Claim~\ref{cl:final-red}]
Let $X_i^*\in \mathcal{I}^*$.
If $G[X_i^*]$ is a connected component of $G$, then because Reduction Rule~\ref{red:comp-del} is not applicable, $|X_i^*|\leq 2$.  Assume that $X_i^*\in\mathcal{I}^*$ is not the set of vertices of a connected component. Then $N_{G^*}(X_i^*)\neq\emptyset$. If $|N_{G^*}(X_i^*)|=1$, then $|X_i^*|=1$, because Reduction Rule~\ref{red:comp-del} cannot be applied. If $|N_{G^*}(X_i^*)|\geq 2$, then $|X_i^*|=1$, because Reduction Rule~\ref{red:comp-two} is not applicable.
In both cases $|X_i^*|=1$ as required. Finally, if $X_i^*\in \mathcal{Z}$, then $|X_i|\leq 3$ because of Reduction Rule~\ref{red:clique-red}. 
\end{proof} 
 
Let $G_1,\ldots,G_s$ be  the connected components of $G^*$ induced by the sets $X_i$ of size two, and let $\hat{G}=G^*-\bigcup_{i=1}^s V(G_i)$, that is, $G^*=\hat{G}+sK_2$. Claim~\ref{cl:final-red} implies that every edge of $\hat{G}$ has at least one end-vertex in a set $X_i$ for $X_i\in\mathcal{Z}$. Since $|X_i|\leq 3$ for each set $X_i\in\mathcal{Z}$, $\vc(\hat{G})\leq 3|\mathcal{Z}|\leq 6k$. Because Reduction Rules~\ref{red:clique-red}--\ref{red:comp-two} are safe, 
the enumeration of   matching cuts (minimal or maximal matching cuts, respectively) for the input graph is equivalent to their enumeration for $G^*=\hat{G}+sK_2$. 
Because $\vc(\hat{G})\leq 6k$, we can apply  Corollary~\ref{cor:vc-kern-A}. Since the initial partition $V(G)$ into the twin classes can be computed in linear time and each of the Reduction Rules~\ref{red:clique-red}--\ref{red:comp-two} can be applied in polynomial time, we conclude that 
 \probMinMC admits a \fp{} enumeration kernel and \probMC and \probMaxMC admit 
 \pd{} enumeration kernels with $\Oh(k^2)$ vertices.
\end{proof}

\section{Enumeration Kernels for the Neighborhood Diversity and Modular Width Parameterizations}\label{sec:nd}
The notion of the \emph{neighborhood diversity} of a graph was introduced by Lampis~\cite{Lampis12} (see also~\cite{Ganian12}). Recall that a  set of vertices $U\subseteq V(G)$ is  a \emph{module} of $G$ if for every vertex $v\in V(G)\setminus U$, either $v$ is adjacent to each vertex of $U$ or $v$ is non-adjacent the vertices of $U$. 
The \emph{neighborhood decomposition} of $G$ is a partition of $V(G)$ into modules such that every module is either a clique or an independent set. We call these modules \emph{clique} or \emph{independent} modules, respectively; note that a module of size one is both clique module and an independent module. 
The \emph{size} of a decomposition is the number of modules. 
The \emph{neighborhood diversity} of a graph $G$, denoted $\nd(G)$, is the minimum size of a neighborhood decomposition. Note (see, e.g., \cite{Lampis12,TedderCHP08}) that the neighborhood diversity and the corresponding neighborhood decomposition can be computed in linear time. We show \fp{} (\pd{}) enumeration kernels for the matching cut problem parameterized by the neighborhood diversity. 

There are many similarities between the results in this subsection and Subsection~\ref{sec:vc}. Hence, we will only sketch some proofs. 

Let $G$ be a graph and let $k=\nd(G)$.  
The case when $G$ has no edges is trivial and can be easily considered separately. From now, we assume that $G$ has at least one edge. 

Consider a minimum-size neighborhood decomposition $\mathcal{U}=\{U_1,\ldots,U_k\}$ of $G$. Let $\mathcal{G}$ be the \emph{quotient} graph for $\mathcal{U}$, that is, $\mathcal{U}$ is the set of nodes of $\mathcal{G}$ and two distinct nodes $U_i$ and $U_j$ are adjacent in $\mathcal{G}$ if and only if the vertices of modules $U_i$ and $U_j$ are adjacent in $G$.  We call the elements of $V(\mathcal{G})$ \emph{nodes} to distinguish them form the vertices of $G$. 
We say that a module $U_i$ is \emph{trivial} if $U_i$ is an independent module and $U_i$ is an isolated node of $\mathcal{G}$. Notice that $\mathcal{U}$ can contain at most one trivial module. We call $U_i$ a \emph{pendent} module if $U_i$ is an independent module of degree one in $\mathcal{G}$ such that its unique neighbor $U_j$ in $\mathcal{G}$ has size one; we say that $U_j$ is a \emph{subpendant} module. Notice that each subpendant is adjacent  to exactly one pendant module and the pendant modules are pairwise nonadjacent in $\mathcal{G}$.  

As in Section~\ref{sec:vc}, our kernelization algorithm is the same for all the considered problems. First, we \emph{mark} some vertices.
\begin{itemize}
\item[(i)] If $\mathcal{U}$ contains a trivial module, then \emph{mark} one its vertex. 
\item[(ii)] For every pendant module $U_i$, \emph{mark} an arbitrary vertex of $U_i$.
\item[(iii)] For ever module $U_i$ that is not trivial of pendant, \emph{mark} arbitrary $\min\{3,|U_i|\}$ vertices. 
\end{itemize}
Let $W$ be the set of marked vertices. Note that since we marked at most three vertices in each module, $|W|\leq 2k$.
We define $H=G[W]$ and our kernelization algorithm returns $H$. 

To see the relations between matching cuts of $G$ and $H$, we show the analog of Lemma~\ref{lem:kernel-vc}. For this, denote by $G'$ the graph obtained by the deletion of the vertices of trivial and pendant modules. The graph $H'$ is obtained from $H$ in the same way. 
\begin{lemma}\label{lem:kernel-nd}
A set of edges $M\subseteq E(G')$ is a matching cut of $G'$ if and only if $M\subseteq E(H')$ and $M$ is a matching cut of $H'$.
\end{lemma}  
\begin{proof}
Suppose that $M\subseteq E(G')$ is a matching cut of $G'$ and assume that $M=E_{G'}(A,B)$ for a partition $\{A,B\}$ of $V(G')$. We show that $M\subseteq E(H')$. 
For the sake of contradiction, assume that there is  $uv\in M$ with $u\in A$ and $v\in B$ such that $uv\notin E(H')$. This means that  $u\notin V(H')$ or $v\notin V(H')$. By symmetry, we can assume without loss of generality that $u\notin V(H')$. Then there is $i\in\{1,\ldots,k\}$ such that $u\in U_i$. Since $U_i$ is not trivial and not pendant, $U_i$ contains three marked vertices $u_1,u_2,u_3$. 
Notice that $v\notin U_i$, because, otherwise, $U_i$ would be a clique module and no matching cut can separate vertices of a clique of size at least 3. Observe also that $u_1,u_2,u_3\in B$ as, otherwise, $v$ would have at least two neighbors in $A$. 
This implies that $U_i$ is an independent module, because $u$ would have at least three neighbors in $B$ otherwise.  
Suppose that $u$ has a neighbor $w\neq v$ in $G$. Then because $M$ is a matching cut, $w\in A$. However, $u_1,u_2,u_3\in B$ are adjacent with $w$, because these vertices are in the same module with $u$. This contradiction implies that $v$ is the unique neighbor of $u$ in $G$. Hence, $v$ is the unique neighbor of every vertex of $U_i$. This means that $U_i$ is a pendant module and $\{v\}$ is the corresponding subpendant module.  However, $G'$ does not contain vertices of the pendant modules by the construction.
 This final contradiction proves that $uv\in E(H')$. Therefore,  $M\subseteq E(H')$.  
Since $H'$ is an induced subgraph of $G'$, $M$ is a matching cut of $H'$.

For the opposite direction, assume that $M$ is a matching cut of $H'$.  Let $M=E_{H'}(A,B)$ for a partition $\{A,B\}$ of $V(H')$.   We claim that for every $v\in V(G')\setminus V(H')$, either $v$ has all its neighbors in $A$ or $v$ has all its neighbors in $B$. For the sake of contradiction, assume that there is $v\in V(G')\setminus V(H')$ that has a neighbor $u\in A$ and a neighbor $w\in B$. Note that $v\in U_i$ for some $i\in\{1,\ldots,k\}$ and $|U_i|\geq 4$. Then $U_i$ contains three marked vertices
and either at least two of these marked vertices of are in $A$ or at least two of these marked vertices  are in $B$. 
In the first case, $w$ has at least two neighbors in $A$, and $u$ has at least two neighbors in $B$ in the second. In both cases, we have a contradiction with the assumption that $M$ is a matching cut. Since for every $v\in V(G')\setminus V(H')$, either $v$ has all its neighbors in $A$ or $v$ has all its neighbors in $B$, $M$ is a matching cut of $G'$.
\end{proof}

To proceed,  we denote by $X$ the set of vertices of the subpendant modules. By the definition of pendant and subpendant modules, for each $x\in X$, $\{x\}$ is a subpendant module that is adjacent in $\mathcal{G}$ to a unique pendant module $U$. We define  $L_x=\{xu\mid u\in U\}$. Notice that for each $x\in X$, $H$ contains a unique edge $\ell_x\in L_x$, because exactly one vertex of $U$ is marked. Let $L=\bigcup_{x\in X}L_x$. Exactly as in Section~\ref{sec:vc}, we say that 
two sets of edges $M_1$ and $M_2$ of $G$ are \emph{equivalent} if $M_1\setminus L=M_2\setminus L$ and 
for every $x\in X$, $|M_1\cap L_x|=|M_2\cap L_x|$. Using exactly the same arguments as in the proof of Lemma~\ref{lem:equiv}, we show its analog.
\begin{lemma}\label{lem:equiv-nd}
A set of edges $M\subseteq E(G)$ is a  matching cut (minimal or maximal matching cut, respectively) of $G$ if and only if $H$ has a matching cut (minimal or maximal matching cut, respectively) $M'$   equivalent to $M$. 
\end{lemma}
The lemma allows us to prove the main theorem of this section by the same arguments as in Theorem~\ref{thm:vc-kern}. 
\begin{theorem}\label{thm:nd-kern}
\probMinMC admits a \fp{} enumeration kernel and \probMC and \probMaxMC admit  
\pd{} enumeration kernels with at most $3k$  vertices when parameterized by the neighborhood  diversity $k$ of the input graph. 
\end{theorem}
Note that similarly to the parameterization by the vertex cover number, \probMaxMC and \probMC do not admit \fp{} enumeration kernels for the parameterization by the neighborhood diversity as demonstrated by the example when $G$ is the union of stars. Observe that the neighborhood diversity of the disjoint union of $k$ stars $K_{1,n/k-1}$ is $2k$.  

Combining Theorems~\ref{thm:nd-kern},~\ref{thm:general} and~\ref{obs:natural}, we obtain the following corollary.
\begin{corollary}\label{thm:upper-nd}
The minimal matching cuts of an $n$-vertex graph $G$ can be enumerated in $2^{\Oh(\nd(G))}\cdot n^{\Oh(1)}$ time, and the (maximal) matching cuts of $G$ can be enumerated with 
$2^{\Oh(\nd(G))}\cdot n^{\Oh(1)}$ delay.
\end{corollary}  
Observe that the neighborhood diversity of $P_n$ is $n$. By Observation~\ref{obs:path}, $P_n$ has $F(n+1)-1=F(\nd(P)+1)-1$ matching cuts. 
This immediately implies that the exponential dependence on $\nd(G)$ in the running time for \probMinMC is unavoidable. 

We conclude this section by considering the parameterization by the \emph{modular width} that is weaker than the neighborhood diversity parameterization but stronger than the cliquewidth parameterization.

 The \emph{modular width} of a graph $G$ (see, e.g.,~\cite{GajarskyLO13}), denoted by $\mw(G)$, is the minimum positive integer $k$ such that a graph isomorphic to $G$ can be recursively constructed by the following operations:
 \begin{itemize}
 \item Constructing a single vertex graph.
 \item The \emph{substitution} operation with respect to some graphs $Q$ with $2\leq r\leq k$ vertices $v_1,\ldots,v_r$ applied to $r$ disjoint graphs $G_1,\ldots,G_r$ of modular width at most $k$; the substitution operation, that generalizes the disjoint union and complete join, creates the graph $G$ with $V(G)=\bigcup_{i=1}^rV(G_i)$ and 
 \begin{equation*}
 E(G)=\big(\bigcup_{i=1}^rE(G_i)\big)\cup\big(\bigcup_{v_iv_j\in E(Q)}\{xy\mid x\in V(G_i)\text{ and }y\in V(G_j) \}\big).
 \end{equation*}
 \end{itemize}
 The modular width of a graph $G$ can be computed in polynomial (in fact, linear) time~\cite{GajarskyLO13,TedderCHP08}.
Notice that $\cw(G)-1\leq \mw(G)\leq \nd(G)$. 

We show that \probMinMC admits a \fp{} enumeration kernel for the modular width parameterization.
\begin{theorem}\label{thm:mw-kern}
\probMinMC admits a \fp{} enumeration kernel with at most $6k$~vertices  when parameterized by the modular width $k$ of the input graph. 
\end{theorem}
\begin{proof}
Let $G$ be a graph with $\mw(G)=k$. If $G$ is disconnected, then the empty set is a unique matching cut of $G$. Then the kernelization algorithm outputs $H=2K_1$. The solution-lifting algorithm is trivial in this case. The case $G=K_1$ is also trivial. 
 Assume from now that $G$ is a connected graph with at least two vertices. This implies that $G$ is obtained by the substitution operation with respect to some connected graph $Q$ with at most $k$ vertices from graphs of modular width at most $k$. This means that $G$ has a \emph{modular decomposition} $\mathcal{U}=\{U_1,\ldots,U_r\}$ for $2\leq r\leq k$, that is, a partition of $V(G)$ into $r$ modules. These modules can be computed in linear time~\cite{TedderCHP08}. For each $i\in \{1,\ldots,r\}$, let $X_i\subseteq U_i$ be the set of isolated vertices of $G[U_i]$ and let $Y_i=U_i\setminus X$. We show the following claim.
\begin{claim}\label{cl:module} 
If $M=E(A,B)$ is a matching  cut of $G$ for a partition $\{A,B\}$ of $V(G)$, then for every $i\in\{1,\ldots,r\}$, either $Y_i\subseteq A$ or $Y_i\subseteq B$.
\end{claim}
\begin{proof}[Proof of Claim~\ref{cl:module}]
Consider $Y_i$ for some $i\in \{1,\ldots,r\}$. The claim is trivial if $Y_i=\emptyset$. Assume that $Y_i\neq\emptyset$. Recall that $H=G[Y_i]$ has no isolated vertices by definition. Let $uv\in E(H)$. Because $G$ is connected and has at least two modules, there is $w\in V(G)\setminus U_i$ that is adjacent to $u$ and $v$. Because $u$, $v$, and $w$ compose a triangle in $G$ that cannot be separated by a matching cut, we conclude that either $u,v,w\in A$ or $u,v,w\in B$. This implies that for every connected component $H'$ of $H$, either $V(H')\subseteq A$ or $V(H')\subseteq B$. Now suppose that $H'$ and $H''$ are distinct components of $H$. Let $uv\in E(H')$ and $u'v'\in E(H'')$. By the same arguments as before, there is  $w\in V(G)\setminus U_i$ that is adjacent to $u$, $v$, $u'$, and $v'$. Since the triangles $uvw$ and $u'v'w$ are either in $A$ or in $B$, we conclude that either 
$V(H')\cup V(H'')\subseteq A$ or $V(H')\cup V(H'')\subseteq B$. Therefore, either $Y_i\subseteq A$ or $Y_i\subseteq B$.
\end{proof}
We construct $G'$ from $G$ by making each set $Y_i$ a clique by adding edges. By Claim~\ref{cl:module}, $M$ is a matching cut of $G$ if and only if $M$ is a matching cut of $G'$. Hence, the enumeration of minimal matching cuts of $G$ is equivalent to the enumeration of minimal matching cuts of $G'$. Because every $X_i$ is an independent set and every $Y_i$ is a clique, $\nd(G)\leq 2r\leq 2k$. This allows to apply Theorem~\ref{thm:nd-kern} that implies the existence of a \fp{} enumeration kernel for \probMinMC with at most $6k$ vertices.  
\end{proof}
Notice that it is crucial for the \fp{} enumeration kernel for \probMinMC parameterized by the modular width that the empty set is the unique matching cut of a disconnected graph. If we exclude empty matching cuts, then we can obtain the following kernelization conditional lower bound. 
\begin{proposition}\label{prop:no-kern-mw}
 The problem of enumerating nonempty matching cuts (minimal or maximal matching cuts, respectively)
  does not admit a \pd{} enumeration kernel of polynomial size when parameterized by the modular width of the input graph unless $\compass$. 
 \end{proposition}
 \begin{proof}
 As with Proposition~\ref{prop:no-kern-tw}, it is sufficient to show that the decision version of the matching cut problem does not admit a polynomial kernel when parameterized by the modular width of the input graph unless $\compass$.
 Let $G_1,\ldots,G_t$ be disjoint graphs of modular width at most $k\geq 2$. Let $G$ be the disjoint union of $G_1,\ldots,G_t$. By the definition of the modular width, $\mw(G)\leq k$.  
 Clearly, $G$ has a nonempty matching cut if and only if there is $i\in\{1,\ldots,t\}$ such that $G_i$ has a matching cut. Since deciding whether a graph has a nonempty matching cut is \classNP-complete~\cite{Chvatal84}, the results of Bodlaender et al.~\cite{BodlaenderDFH09} imply that the decision problem does not admit a polynomial kernel unless $\compass$.   
  \end{proof}
Proposition~\ref{prop:no-kern-mw} indicates that it is unlikely that \probMC and \probMaxMC have polynomial 
\pd{} enumeration 
kernels of polynomial size under the modular width parameterization. Notice, however, that Proposition~\ref{prop:no-kern-mw} by itself does not imply a kernelization lower bound.

\section{Enumeration Kernels for the Parameterization by the Feedback Edge Number}\label{sec:fen}
A set of edges $X$ of a graph $G$ is said to be a \emph{feedback} edge set if $G-S$ has no cycle, that is, $G-S$ is a forest. The minimum size of a feedback edge set is called the \emph{feedback edge} number or the \emph{cyclomatic} number. We use $\fen(G)$ to denote the feedback edge number of a graph $G$. It is well-known (see, e.g.,~\cite{Diestel12}) that if $G$ is a graph with $n$ vertices, $m$ edges and $r$ connected components, then $\fen(G)=m-n+r$ and a feedback edge set of minimum size can be found in linear time. Throughout this section, we assume that the input graph in an instance of \probMinMC or \probMC is given together with a feedback edge set.  Equivalently, we may assume that kernelization and solution-lifting algorithms are supplied by the same algorithm computing a minimum feedback edge set. Then this algorithm computes exactly the same set for the given input graph.
Note that $\tw(G)\leq \fen(G)+1$, because a forest can be obtained from~$G$ by deleting an arbitrary end-vertex of each edge of a feedback edge set.

Our algorithms use the following folklore observation that we prove for completeness.

\begin{observation}\label{obs:vertices}
Let $F$ be a forest. 
Let also $n_{\leq 1}$ be the number of vertices of degree at most one, $n_2$ be the number of vertices of degree two, and $n_{\geq 3}$ be the number of vertices of degree at least three.  Then $n_{\geq 3}\leq n_{\leq 1}-2$.    
\end{observation} 

\begin{proof}
Denote by $n_0$ the number of isolated vertices, and let $n_1$ be the number of vertices of degree one.
 Observe that $|V(F)|=n_0+n_1+n_2+n_{\geq 3}$, $|E(F)|\leq |V(F)|-1-n_0$, and 
 \begin{equation*}
 2|E(F)|=\sum_{v\in V(F)}d_F(v)\geq n_1+2n_2+3n_{\geq 3}.
 \end{equation*}
 Then 
 \begin{equation*}
 n_1+2n_2+3n_{\geq3}\leq 2(n_{1}+n_2+n_{\geq 3})-2-2n_0.
 \end{equation*} 
Therefore,  $n_{\geq 3}\leq n_{1}-2-2n_0\leq n_{\leq 1}-2$.
\end{proof}

In contrast to vertex cover number and neighborhood diversity, \probMinMC does not admit a \fp{} enumeration kernel in case of the feedback edge number: let $\ell$~and~$k$ be positive integers and consider the graph $H_{k,\ell}$  that is constructed as follows. 
\begin{itemize}
\item For every $i\in\{1,\ldots,k\}$, construct two vertices $u_i$ and $v_i$ and a $(u_i,v_i)$-path of length $\ell$.
\item Add edges to make each of $u_1,\cdots,u_k$ and~$v_1,\cdots,v_k$ a path of length~$k-1$.
\end{itemize} 
Observe that $H_{k,\ell}$ has at least $\ell^k$ minimal matching cuts composed by taking one edge from  every $(u_i,v_i)$-path. Since $H_{k,\ell}$ has $n=k(\ell+1)$ vertices and $\fen(H_{k,\ell})=k-1$, the number of minimal matching cuts is at least $\big(\frac{n}{\fen(H_{k,\ell})-1}-1\big)^{\fen(H_{k,\ell})}$. This immediately implies that the minimal matching cuts cannot be enumerated in \classFPT time. In particular,  \probMinMC cannot have a \fp{} enumeration kernel by Theorem~\ref{obs:natural}. However, this problem and \probMC admit \pd{} enumeration kernels. 

The kernels for the problems are similar but the kernel for \probMC  requires some technical details that do not appear in the kernel for \probMinMC. Therefore, we consider \probMinMC  separately even if some parts of the proof of the following theorem will be repeated later.

\begin{theorem}\label{thm:fen-kern-min}
\probMinMC admits a 
\pd{} enumeration kernel with $\Oh(k)$ vertices when parameterized by the feedback edge number $k$ of the input graph.
\end{theorem}

\begin{proof}
Let $G$ be a graph with $\fen(G)=k$ and a feedback edge set $S$ of size $k$. If $G$ is disconnected, then the empty set is the unique minimal matching cut. Accordingly, the kernelization algorithm returns $2K_1$ and the solution-lifting  algorithm outputs the empty set for the empty matching cut of $2K_1$. If  $G=K_1$, then the kernelization algorithm simply returns $G$.  If $G$ is a tree with at least one edge, then the kernelization algorithm returns $K_2$. Then for the unique matching cut of $K_2$, the solution-lifting algorithm outputs $n-1$ minimal matching cuts of $G$ composed by single edges. Clearly, this can be done in $\Oh(n)$ time. We assume from now that $G$ is a connected graph distinct from a tree, that is, $S\neq\emptyset$.  

If $G$ has a vertex of degree one, then pick an arbitrary such vertex $u^*$. Let $e^*$ be the edge incident to $u^*$. We iteratively delete vertices of degree at most one distinct from $u^*$. Denote by $G'$ the obtained graph. Notice that $G'$ has at most one vertex of degree one and if such a vertex exists, then this vertex is $u^*$. Observe also that $S$ is a minimum feedback edge set of $G'$. Let $T=G'-S$. Clearly $T$ is a tree. Notice that $T$ has at most $2|S|+1\leq 2k+1$ leaves. By Observation~\ref{obs:vertices}, $T$ has at most $2k-1$ vertices of degree at least three. Denote by $X$ the set of vertices of $T$ that either 
are end-vertices of the edges of $S$, or have degree one,  
or have degree at least three. Because every vertex of $T$ of degree one is either $u^*$ or an end-vertex of some edge of $S$,
we have that $|X|\leq 4k$. 
By the construction, every vertex $v$ of $G'$ of degree two is an inner vertex of an $(x,y)$-path $P$ such that $x,y\in X$ and the inner vertices of $P$ are outside $X$. Moreover, for every two distinct $x,y\in X$, $G'$ has at most one $(x,y)$-path $P_{xy}$ with all its inner vertices outside $X$. We denote by $\mathcal{P}$ the set of all such paths.
We say that an edge of $P_{xy}$ is the \emph{$x$-edge} if it is incident to $x$ and is the \emph{$y$-edge} if is incident to $y$; the other edges are said to be \emph{middle} edges of $P_{xy}$. We say that $P_{xy}$ is \emph{long}, if $P_{xy}$ has length at least four. Then we apply the following reduction rule exhaustively. 

\begin{reduction}\label{red:contr} 
If there is a long path $P_{xy}\in \mathcal{P}$ for some $x,y\in X$, then contract an arbitrary middle edge of $P_{xy}$.
\end{reduction}

\noindent
Denote by $H$ the obtained graph. Denote by $\mathcal{P}'$ the set of paths obtained from the paths of $\mathcal{P}$; we use $P_{xy}'$ to denote the path obtained from $P_{xy}\in\mathcal{P}$. 

Our kernelization algorithm returns $H$ together with $S$. 

Recall that $S\subseteq E(G[X])$. Then $H-E(G[X])$ is a forest. Moreover, by the construction, the vertices of degree at most one of this forest are in $X$.
 This implies that $|\mathcal{P}'|\leq |X|-1$.
Because  $|X|\leq 4k$, $|\mathcal{P}'|\leq 4k-1$. 
Since $\mathcal{P}'$ does not contain long paths, every path of $\mathcal{P}'$ contains at most two inner vertices. 
Therefore, $|V(H)|\leq |X|+2|\mathcal{P}'|\leq 4k+2(4k-1)\leq 10k$. This means that $H$ has the required size.

To construct the solution-lifting algorithm, we need some properties of minimal matching cuts of $G$. Observe that the set $\mathcal{M}$ of all minimal matching cuts of $G$ can be partitioned into three (possibly empty) subsets $\mathcal{M}_1$, $\mathcal{M}_2$, and $\mathcal{M}_3$ as follows.
\begin{itemize}
\item Every edge of $E(G)\setminus E(G')$ is a bridge of $G$ and, therefore, forms a minimal matching cut of $G$. We define $\mathcal{M}_1$ to be the set of these matching cuts, that is, $\mathcal{M}_1=\{\{e\}\mid e\in E(G)\setminus E(G')\}$.
\item For every $P_{xy}\in \mathcal{P}$, every minimal matching cut contains at most two edges of $P_{xy}$. Moreover, every two edges of $P_{xy}$ with distinct end-vertices form a minimal matching cut of $G$, unless the edges of $P_{xy}$ are bridges of $G$. We define 
$\mathcal{M}_2=\{\{e_1,e_2\}\mid \{e_1,e_2\}\text{ is a minimal matching cut  of }G\text { such that }e_1,e_2\in P_{xy}\text{ for some }P_{xy}\in\mathcal{P}\}$.
\item The remaining minimal matching cuts compose $\mathcal{M}_3$. Notice that for every matching cut $M\in\mathcal{M}_3$, (i) $M\subseteq E(G')$ and $M$~is a minimal matching cut of $G'$, and (ii) for every $P_{xy}\in\mathcal{P}$, $|M\cap E(P_{xy})|\leq 1$.
\end{itemize}
We use this partition of $\mathcal{M}$ in our solution-lifting algorithm. For this, we define $\mathcal{M}'$ to be the set of minimal matching cuts of $H$. We also consider the partition of $\mathcal{M}'$ into $\mathcal{M}_2'$ and $\mathcal{M}_3'$, where $\mathcal{M}'_2$ is the set of all minimal matching cuts of $H$ formed by two edges of some $P_{xy}'\in\mathcal{P}'$ and $\mathcal{M}_3'=\mathcal{M}'\setminus \mathcal{M}_2'$. 
Similarly to $\mathcal{M}_3$, we have that  for every $M\in\mathcal{M}_3'$, $M$ is a minimal matching cuts of $H$ and for every $P_{xy}'\in\mathcal{P}'$, $|M\cap E(P_{xy}')|\leq 1$.
Observe that, contrary to $\mathcal{M}$, 
$\mathcal{M}'$ is partitioned into two sets as $\mathcal{M}'$ does not contain matching cuts corresponding to the cuts of $\mathcal{M}_1$.

Notice that by our assumption the input graph is given together with $S$ and $S\subseteq E(H)$. This allows us to find $u^*$ and $e^*$ (if they exist) as $u^*$ is the unique vertex of degree one in $G'$. Then we can recompute $X$ and the sets of paths $\mathcal{P}$ and $\mathcal{P}'$ of $G$ and $H$, respectively, in polynomial time. Hence, we can assume that the solution-lifting algorithm has access to these sets. 

First, we deal with $\mathcal{M}_1$. Notice that if $\mathcal{M}_1\neq\emptyset$, then $G$ has at least one vertex of degree one. This means, that $H$ contains $u^*$ and $e^*$. Recall that $u^*$ is a vertex of degree one and $e^*$ is the edge incident to $u^*$. Observe that $\{e^*\}$ is a minimal matching cut in both $G$ and $H$, and $\{e^*\}\in\mathcal{M}_3$ and $\{e^*\}\in\mathcal{M}_3'$.  Given the minimal matching cut $\{e^*\}$ of $H$, the solution-lifting algorithm outputs $\{e^*\}$ and then the minimal matching cuts of $\mathcal{M}_1$. Clearly, $|\mathcal{M}_1|\leq n$ and the elements of $\mathcal{M}_1$ can be listed with constant delay.

Next, we consider $\mathcal{M}_2$. If $\mathcal{M}_2\neq\emptyset$, then there is $P_{xy}\in\mathcal{P}$ of length at least three, such that $\{e_1,e_2\}$ is a minimal matching cut for $G$ for some $e_1,e_2\in E(P_{xy})$. Notice that the corresponding path $P_{xy}'\in \mathcal{P}'$ in $H$ has length three and the $x$- and $y$-edges of $P_{xy}'$ form a minimal matching cut  of $H$. Moreover, this is the unique minimal matching cut of $H$ formed by two edges of $P_{xy}'$. 
Given a minimal matching cut $\{e_1,e_2\}\in \mathcal{M}_2'$ of $H$ such that $e_1$ and $e_2$ are $x$- and $y$-edges of some path $P_{xy}'\in \mathcal{P}'$, the solution-lifting algorithm outputs the matchings $\{e_1',e_2'\}$, where $e_1',e_2'\in E(P_{xy})$. Notice that we have at most $n^2$ such matchings and they can be enumerated with polynomial delay. It is straightforward to verify that for every minimal matching cut of $\mathcal{M}_2'$, the solution-lifting algorithm outputs a nonempty set of minimal matching cuts of $G$, the matching cuts listed for distinct element of $\mathcal{M}_2'$ are distinct,
and the union of all produced minimum matching cuts over all elements of $\mathcal{M}_2'$ gives $\mathcal{M}_2$.

Finally, we consider $\mathcal{M}_3$. Recall that for every $M\in \mathcal{M}_3$, $M$  is a minimal matching cut of $G'$ and $|M\cap E(P_{xy})|\leq 1$ for every $P_{xy}\in\mathcal{P}$. 
 Let $M\in\mathcal{M}_3$ and $M'\in\mathcal{M}_3'$. We say that $M$ and $M'$ are \emph{equivalent}  
if $M\cap E(G[X])=M'\cap E(G[X])$ and the following holds for every $P_{xy}\in\mathcal{P}$:
\begin{itemize}
\item $M$ contains the $x$-edge of $P_{xy}$ if and only if $M'$ contains the $x$-edge of $P_{xy}'$,
\item $M$ contains the $y$-edge of $P_{xy}$ if and only if $M'$ contains the $y$-edge of $P_{xy}'$,
\item $M$ contains a middle edge of $P_{xy}$ if and only if $M'$ contains the unique middle edge of $P_{xy}'$.
\end{itemize}
By the construction of $H$, it is straightforward to see that a minimal matching cut $M$ of $G$ is in $\mathcal{M}_3$ if and only if there is an equivalent minimal matching cut $M'\in\mathcal{M}_3'$ of $H$. Notice also that if $M_1',M_2'\in \mathcal{M}_3'$ are distinct, then any $M_1,M_2\in\mathcal{M}_3$ that are equivalent to $M_1'$ and $M_2'$, respectively, are distinct. 

Recall that if $G$ has a vertex of degree one, then the vertex $u^*$ and the edge $e^*$ are in $G'$ and $H$, and it holds that . 
 $\{e^*\}$ is a minimal matching cut of both $G$ and $H$ that belongs to $ \mathcal{M}_3$ and $\mathcal{M}_3'$. Note that $\{e^*\}$ is the unique matching cut in $\mathcal{M}_3$ that is equivalent to $\{e^*\}$. Recall that we already explained the output of the solution-lifting algorithm for $\{e^*\}$. In particular,  the algorithms outputs $\{e^*\}\in \mathcal{M}_3$.  

Assume that $M'\in\mathcal{M}_3'$ is distinct from $\{e^*\}$ (or $e^*$ does not exist). The solution-lifting algorithm lists all minimal matching cuts $M$ of $\mathcal{M}_3$ that are equivalent to $M'$. For this, we use the recursive branching algorithm 
\textsc{Enum Equivalent} that takes as an input a matching $L$ of $G$ and a path set~$\mathcal{R}\subseteq \mathcal{P}$ and outputs the equivalent matching cuts $M'$ of $G$ such that 
(i) $L\subseteq M'$, (ii) $M'$ is equivalent to $M$,  and (iii) the constructed matchings $M'$  differ only by some edges of the paths $P_{xy}\in\mathcal{R}$. In other words, the algorithm extends the partial matching cut by adding edges from the path set $\mathcal{R}$. 
To initiate the computations, we construct the initial matching $L'$ of $G$ and the initial set of paths $\mathcal{R}'\subseteq \mathcal{P}$ as follows.
First, we set $L':=L'\cap E(G[X])$ and $\mathcal{R}'=\emptyset$. Then for each $P_{xy}\in \mathcal{P}$ we do the following:
\begin{itemize}
\item if the $x$- or $y$-edge $e$ of $P_{xy}'$ is in $M'$, then set $L'=L'\cup\{e\}$,
\item if the middle edge of $P_{xy}'$ is in $M'$, then set $\mathcal{R}':=\mathcal{R}'\cup \{P_{xy}\}$.   
\end{itemize} 
Observe that by the definition of equivalent matching cuts, a minimal matching cut $M$ is equivalent to $M'$ if and only if $M$ can be constructed from $L'$ by adding one middle edge of every path $P_{xy}\in \mathcal{R}'$. Then calling $\textsc{Enum Equivalent}(L',\mathcal{R}')$ solves the enumeration problem.

\begin{algorithm}[H]\label{alg:min-fen}
\caption{$\textsc{Enum Equivalent}(L,\mathcal{R})$}
\If{$\mathcal{R}=\emptyset$}{\Return{$L$} and {\bf quit}}
\ElseIf{$\mathcal{R}\neq\emptyset$}
{
select arbitrary $P_{xy}\in\mathcal{R}$\;
\ForEach{for each middle edge $e$ of $P_{xy}$}
{$\textsc{Enum Equivalent}(L\cup\{e\}, \mathcal{R}\setminus\{P_{xy}\})$
}
}
\end{algorithm}

As Algorithm~\ref{alg:max} in the proof of Theorem~\ref{thm:vc-kern}, this algorithm is a standard backtracking enumeration algorithm. The depth of the recursion is upper-bounded by $n$. This implies that $\textsc{Enum Equivalent}(L',\mathcal{R}')$ enumerates all minimal matching cuts $M\in \mathcal{M}_3$ that are equivalent to $M'$ with polynomial delay.  

Summarizing the considered cases, we obtain that if the edge $e^*$ exists, the solution-lifting algorithm enumerates all minimal matching cuts of $\mathcal{M}_1$ and $\{e^*\}$, and if $e^*$ does not exist, then $\mathcal{M}_1$ is empty. Then for every minimal matching cut $M'\in\mathcal{M}_2'$, the solution-lifting algorithm enumerates the corresponding minimal matching cuts of $\mathcal{M}_2$; the minimal matching cuts of $G$ generated for distinct minimal matching cuts of $H$ are distinct, and every minimal matching cut in $\mathcal{M}_2$ is generated for some minimal matching cut from $\mathcal{M}_2'$. Finally, for every minimal  matching cut $M'\in \mathcal{M}_3'$ distinct from $\{e^*\}$, we enumerate equivalent minimal matching cuts of $G$. In this case we also have that  the minimal matching cuts of $G$ generated for distinct minimal matching cuts of $H$ are distinct, and every minimal matching cut in $\mathcal{M}_3$ is generated for some minimal matching cut from $\mathcal{M}_3'$. We conclude that the solution-lifting algorithm satisfies condition (ii$^*$) of the definition of a 
 \pd{} enumeration kernel.
\end{proof}

For \probMC, we need the following  observation that  follows from the results of Courcelle~\cite{Courcelle09} similarly to Proposition~\ref{prop:tw-cw}. 
For this, we note that the results of~\cite{Courcelle09} are shown for the extension of MSOL called \emph{counting monadic second-order logic} (CMSOL). For every integer constants $p\geq 1$ and $q\geq 0$, CMSOL 
includes a predicate ${\sf Card}_{p,q}(X)$ for a set variable $X$ which tests whether $|S|\equiv q\pmod{p}$. Also we can apply the results for labeled graphs whose vertices and/or edges have labels from a fixed finite set.

\begin{observation}\label{obs:fixed}
Let $F$ be a forest and let $A,B,C\subseteq E(F)$ be disjoint edge sets. Then all  matchings $M$ of $F$ such that $A\subseteq M$, $B\cap M=\emptyset$, and either $C\subseteq M$ or $C\cap M=\emptyset$  can be enumerated with polynomial delay. Moreover, if $u,v$ are distinct vertices of the same connected component of $F$ and $h\in\{0,1\}$, then all such (nonempty) matchings with  the additional property that  $|E(P)\cap A|\mod 2=h$, where $P$ is the $(u,v)$-path in $F$, also can  be enumerated with polynomial delay. 
\end{observation}

Now we show 
 a polynomial-delay enumeration kernel for \probMC.

\begin{theorem}\label{thm:fen-kern}
 \probMC admits a \pd{} enumeration kernel with $\Oh(k)$ vertices when parameterized by the feedback edge number $k$ of the input graph.
\end{theorem}

\begin{proof}
Let $G$ be a graph with $\fen(G)=k$ and an edge feedback set $S$ of size $k$. It is convenient to consider the case when $G$ is a forest separately. If $G=K_1$, then the kernelization algorithm returns $K_1$ and the solution-lifting algorithm is trivial. If $G$ has at least two vertices, then the kernelization algorithm returns $2K_1$ that has the unique empty matching cut. Then, given this empty matching cut, the solution-lifting algorithm lists all matching cuts of $F$ with polynomial delay using  Observation~\ref{obs:fixed} (or Proposition~\ref{prop:tw-cw}). 
We assume from now that $G$ is not a forest. In particular, $S\neq\emptyset$.
 
If $G$ has one or more connected component that are trees, we select an arbitrary vertex $v^*$ of these components. 
If $G$ has a connected component that contains a vertex of degree one and is not a tree,  then arbitrary select such a vertex $u^*$ of degree one and denote by  $e^*$ be the edge incident to $u^*$. Then we iteratively delete vertices of degree at most one distinct from $u^*$ and $v^*$. Denote by $G'$ the obtained graph. Notice that $G'$ has at most one isolated vertex (the vertex~$v^*$) and at most one vertex of degree one (the vertex $u^*$).
Observe also that $S$ is a minimum feedback edge set of $G'$. Let $T=G'-S$. Clearly $T$ is a forest. Notice that $T$ has at most $2|S|+2\leq 2k+2$ vertices of degree at most one. By  Observation~\ref{obs:vertices}, $T$ has at most $2k$ vertices of degree at least three. 
Denote by $X$ the set of vertices of $T$ that either 
are end-vertices of the edges of $S$, or have degree one,  
or have degree at least three. Because every vertex of $T$ of degree at most one is either $u^*$, $v^*$, or an end-vertex of some edge of $S$,
we have that $|X|\leq 4k+2$.

In the same way as in the proof of Theorem~\ref{thm:fen-kern-min}, we have that every vertex $v$ of $G'$ of degree two is an inner vertex of an $(x,y)$-path $P$ such that $x,y\in X$ and the inner vertices of $P$ are outside $X$. Moreover, for every two distinct $x,y\in X$, $G'$ has at most one $(x,y)$-path $P_{xy}$ with all its inner vertices outside $X$. We denote by $\mathcal{P}$ the set of all such paths.
We say that an edge of $P_{xy}$ is the \emph{$x$-edge} is it is incident to $x$ and is the \emph{$y$-edge} if is incident to $y$. 
 We say that an edge $e$ of $P_{xy}$ is a \emph{second $x$-edge} (a \emph{second $y$-edge}, respectively) if $e$ has a common end-vertex with the $x$-edge (with the $y$-edge, respectively). The edges that are distinct from the $x$-edge, the second $x$-edge, the $y$-edge and the second $y$-edge are called \emph{middle}. 
We say that $P_{xy}$ is \emph{long}, if $P_{xy}$ has length at least six; otherwise, $P_{xy}$ is \emph{short}.
Let $F=G-E(G')$. Since $S\subseteq E(G')$, $F$ is a forest. Moreover, each connected component $T$ of $F$ has at most one vertex in $V(G')$.

We exhaustively apply the following reduction rule. 
\begin{reduction}\label{red:contr-new} 
If there is a long path $P_{xy}\in \mathcal{P}$ for some $x,y\in X$, then contract an arbitrary middle edge of $P_{xy}$.
\end{reduction}
Let $H$ be the graph obtained from $G'$ by the exhaustive application of Reduction Rule~\ref{red:contr-new}. We also denote by $\mathcal{P}'$ the set of paths obtained from the paths of $\mathcal{P}$; we use $P_{xy}'$ to denote the path obtained from $P_{xy}\in\mathcal{P}$. 

Our kernelization algorithm returns $H$ together with $S$. 

 To upper-bound the size of $H$, notice that $H-E(G[X])$ is a forest such that its vertices of degree at most one are in $X$. 
 This implies that $|\mathcal{P}'|\leq |X|-1\leq 4k+1$. Because $H$ has no long paths, each path $P_{xy}'\in \mathcal{P}'$ has at most four inner vertices, and the total number of inner vertices of all the paths of $\mathcal{P}'$ is at most $4|\mathcal{P}'|\leq 16k+1$. Then $|V(H)|\leq |X|+4|\mathcal{P}'|\leq 20k+1$ implying that $H$ has the required size. 

For the construction of the solution-lifting algorithm, recall that by our assumption the input graph is given together with $S$ and $S\subseteq E(H)$. Then we can identify $v^*$, $u^*$ and $e^*$ in $G$ and $H$, and then we can recompute the set $X$. Next, we can compute the sets of paths~$\mathcal{P}$ and $\mathcal{P}'$ of $G$ and $H$, respectively, in polynomial time. This allows us to assume that the solution-lifting algorithm has access to these sets. 

To construct the solution-lifting algorithm, denote by $\mathcal{M}$ and $\mathcal{M}'$ the sets of matching cuts of $G$ and $H$, respectively. Define $\mathcal{M}_1=\{M\in\mathcal{M}\mid M\cap E(G')=\emptyset\}$ and $\mathcal{M}_2=\{M\in\mathcal{M}\mid M\cap E(G')\neq \emptyset\}$. Notice that $M\in\mathcal{M}_1$ is nonempty if and only if $M$ is a nonempty matching of $F=G-E(G')$. 
First, we deal with the matching cuts of $\mathcal{M}_1$. Observe that $G$ is connected if and only if $H$ is connected. This means that the empty set is a matching cut of $G$ if and only if the empty set is a matching cut of $H$.

Suppose that $H$ has the empty matching cut. Then the solution-lifting algorithm, given this matching cut of $H$, outputs the matching cuts of $\mathcal{M}_1$. Notice that $\mathcal{M}_1\neq\emptyset$, because $\mathcal{M}_1$ contains the empty matching cut. 
The solution-lifting algorithm outputs the empty matching cut and all nonempty matchings of $F$ using Observation~\ref{obs:fixed}.  

Assume now that $H$ is connected. Then $G$ is connected as well and $\mathcal{M}_1\not=\emptyset$ if and only if $F\not=\emptyset$. By the construction of $G'$, if $F$ is not empty, then $G$ has a vertex of degree one. In particular, the kernelization algorithm selects $u^*$ and $e^*$ in this case. 
Notice that $e^*$ is a bridge of $G$, and it holds that $\{e^*\}$ is a matching cut of both $G$ and $H$. Observe also that $\{e^*\}\in \mathcal{M}_2$. This matching cut is generated by the solution-lifting algorithm for the cut $\{e^*\}$ of $H$: 
when the algorithm finishes listing the matching cuts of $\mathcal{M}_2$ for $\{e^*\}$,  it switches to the listing of all nonempty matchings of $F$. This can be done with polynomial delay by Observation~\ref{obs:fixed}. 

Next, we analyze the matching  cuts of $\mathcal{M}_2$. By definition, a matching cut $M$ of $G$ is in $\mathcal{M}_2$ if $M\cap E(G')\neq\emptyset$. This means that $M\cap E(G')$ is a matching cut of $G'$, and for a nonempty matching $M$ of $G$, $M\in \mathcal{M}_2$ if and only if $M\cap E(G')$ is a nonempty matching cut of $G'$. We exploit this property and the solution-lifting algorithm lists nonempty matching cuts of $G'$ and then for each matching cut of $G'$, it outputs all its possible extensions by matchings of $F$. For this, we define the following relation between matching cuts of $H$ and~$G'$. 
Let $M$ be a nonempty matching cut of $H$ and let $M'$ be a nonempty matching of $G'$ (note that we do not require $M'$ to be a matching cut). We say that $M'$ is \emph{equivalent} to $M$ if the following holds:
\begin{itemize}
\item[(i)] $M\cap E(H[X])=M'\cap E(G[X])$ (note that $H[X]=G[X]$).
\item[(ii)] For every $P_{xy}\in \mathcal{P}$ such that $P_{xy}$ is short, $M\cap E(P_{xy}')=M'\cap E(P_{xy})$ (note that $P_{xy}=P_{xy}'$ in this case).
\item[(iii)] For every long $P_{xy}\in\mathcal{P}$,
\begin{itemize}
\item[(a)] $M\cap E(P_{xy}')\neq\emptyset$ if and only if $M'\cap E(P_{xy})\neq\emptyset$,
\item[(b)] $|M\cap E(P_{xy}')|\mod 2 =|M'\cap E(P_{xy})|\mod 2$,
\item[(c)] the $x$-edge ($y$-edge, respectively) of $P_{xy}'$ is in $M'$ if and only if  the $x$-edge ($y$-edge, respectively) of $P_{xy}$ is in $M$,
\item[(d)] if for the second $x$-edge $e_x$, the second $y$-edge $e_y$ and the middle edge $e$ of $P_{xy}'$, $|M\cap\{e_x,e_y,e\}|=1$, then
\begin{itemize}
\item[-] $e_x\in M$ ($e_y\in M$, respectively) if and only if $e_x\in M'$ and $e_y\notin M'$ ($e_x\notin M'$ and $e_y\in M'$, respectively), 
\item[-] $e\in M$ if and only if either $e_x,e_y\in M'$ or $e_x,e_y\notin M'$.
\end{itemize}
(note that $e_x,e_y$ are the second $x$-edge and $y$-edge of $P_{xy}$, because $P_{xy}'$ is constructed by contracting of some middle edges of $P_{xy}$).
\end{itemize}  
\end{itemize}  
We use the properties of the relation summarized in the following claim.
\begin{claim}\label{cl:eq-fen}~
\begin{itemize}
\item[\em (i)] For every nonempty matching cut $M$ of $H$, there is a nonempty matching  $M'$ of $G'$ that is equivalent to $M$.
\item[\em (ii)] For every nonempty matching cut $M$ of $H$ and every nonempty matching $M'$ of $G'$ equivalent to $M$, $M'$ is a matching cut of $G'$.
\item[\em (iii)] Every nonempty matching cut $M'$ of $G'$ is equivalent to at most one matching cut of $H$.
\item[\em (iv)] For every nonempty matching cut $M'$ of $G'$, there is a nonempty matching cut of $M$ such that $M'$ is equivalent to $M$.
\end{itemize}
\end{claim}

 \begin{proof}[Proof of Claim~\ref{cl:eq-fen}]
Let $Y=X\cup \{V(P_{xy}')\mid P_{xy}\in\mathcal{P}\text{ is short}\}$. Notice that conditions (i) and (ii) of the definition of the equivalency of $M'$ to $M$ can be written as $M\cap E(H[Y])=M'\cap E(G[Y])$.
 
To show (i), consider a nonempty matching cut $M$ of $H$. We construct $M'$ as follows. First, we include in $M'$ all the edges of $M$ that are in $H[Y]$. Then for every long path $P_{xy}\in \mathcal{P}$ we do the following.
\begin{itemize}
\item If $M$ contains the $x$-edge or the second $x$-edge (the $y$-edge or the second $y$-edge, respectively) $e$ of $P_{xy}'$, then include $e$ in $M'$.
\item If $M$ contains the middle edge of $P_{xy}$, then include in $M'$ an arbitrary  middle edge of $P_{xy}$. 
\end{itemize}
It is straightforward to verify that $M'$ is a matching of $G$ and $M'$ is equivalent to $M$.

To prove (ii), let $M$ be a nonempty matching cut of $H$ and let $M'$ be a matching of $G'$ equivalent to $M$. Since $M'$ is a matching, to show that $M'$ is a matching cut, we have to prove that there is a partition $\{A',B'\}$ of $V(G')$ such that $M'=E_{G'}(A',B')$. Because $M$ is a nonempty matching cut, there is a partition $\{A,B\}$ of $V(H)$ such that $M=E_H(A,B)$. Let $\hat{A}=A\cap Y$ and $\hat{B}=B\cap Y$. Notice that $E(\hat{A},\hat{B})\subseteq M$ and, therefore,
$E(\hat{A},\hat{B})\subseteq M'$. Moreover, for every $xy\in M'$ such that $x,y\in Y$,  the end-vertices of $xy$ are in distinct sets $\hat{A}$ and $\hat{B}$. 
We show that $\{\hat{A},\hat{B}\}$ can be extended to a required partition $\{A',B'\}$ with $\hat{A}\subseteq A'$ and $\hat{B}\subseteq B'$. Initially, we set $A':=\hat{A}$ and $B':=\hat{B}$.

Recall that the vertices of $V(G')\setminus Y$ are internal vertices of long paths $P_{xy}\in \mathcal{P}$. Consider such a path $P_{xy}$. We use the property that the numbers of edges of $M$ in $P_{xy}'$ and of $M'$ in $P_{xy}$ have the same parity. Suppose that $x\in \hat{A}$ and $y\in\hat{B}$. Let $Q_1,\ldots,Q_r$ be the connected components of $P_{xy}'-M$ listed with respect to the path order in which they occur in $P_{xy}'$ starting from $x$. Then $r$ is even and $V(Q_i)\subseteq A$ if $i$ is odd and $V(Q_i)\subseteq B$ if $i$ is even. Therefore, $|M\cap V(P_{xy}')|$ is odd. Then $|M'\cap V(P_{xy}')|$ is odd, and for the connected components $R_1,\ldots,R_s$ of $P_{xy}-M'$, $s$ is even. Assume that the connected components are listed with respect to the path order induced by $P_{xy}$. We modify $A'$ by setting $A':=A'\cup V(R_i)$ for odd $i\in\{1,\ldots,s\}$ and $B':=B'\cup V(R_i)$ for even $i\in\{1,\ldots,s\}$. By this construction, 
$E_{G'}(V(P_{xy})\cap A',V(P_{xy})\cap B')=M'\cap E(P_{xy})$.  The case when $x$ and $y$ are in the same set $\hat{A}$ or $\hat{B}$, respectively, is analyzed by the same parity arguments. We conclude that by going over all long paths $P_{xy}\in\mathcal{P}$, we construct $A'$ and $B'$ such that $M'=E(A',B')$. This completes the proof of (ii).

From now, we can assume that each nonempty matching of $G'$ that is equivalent to a nonempty matching cut of $H$ is a matching cut.

We show (iii) by contradiction. Assume that there are two distinct nonempty matching cuts $M_1$ and $M_2$ of $H$ and a nonempty matching cut $M'$ of $G'$ such that $M'$ is equivalent to both $M_1$ and $M_2$. The definition of equivalency implies that there is a long path $P_{xy}\in \mathcal{P}$ such that $M_1\cap E(P_{xy}')\neq M_2\cap E(P_{xy}')$. Notice that both sets are nonempty and 
the sizes of both sets have the same parity. We consider two cases depending on the parity. Let $e_x$ denote the $x$-edge, $e_x'$ the second $x$-edge, $e_y$ the $y$-edge, $e_y'$ the second $y$-edge, and $e$ the unique middle edge of $P_{xy}'$. 

Suppose that $|M_1\cap E(P_{xy}')|$ is odd. If $|M_1\cap E(P_{xy}')|=3$, then $M_1\cap E(P_{xy}')=\{e_x,e,e_y\}$, because  $\{e_x,e,e_y\}$ is the unique matching of $P_{xy}'$ of size three.
Since $M_1\cap E(P_{xy}')\neq M_2\cap E(P_{xy}')$, we obtain that $|M_2\cap E(P_{xy}')|=1$. In particular, either $e_x\notin M_2$ or $e_y\notin M_2$. Assume by symmetry, that $e_x\notin M_2$. However, by condition (iii)(a) of the definition of equivalency, the $x$-edge of $P_{xy}$ is in $M'$ if and only if $x$-edge of $P_{xy}'$ in $M_1$ and $M_2$; a contradiction.   Hence, we can assume that $|M_1\cap E(P_{xy}')|=|M_2\cap E(P_{xy}')|=1$. If either $M_1\cap P_{xy}'$ or $M_2\cap E(P_{xy}')$ consists of the $x$-edge or the $y$-edge of $P_{xy}'$, we use the same arguments as above. This means that both $M_1$ and $M_2$ contains a unique edge of $P_{xy}'$ that belongs to $\{e_x',e,e_y'\}$. Suppose that $e_x'\in M_1$ but $e_x'\notin M_2$. However, this contradicts  (iii)(d).
By symmetry, we obtain that $e_x'\notin M_1,M_2$ and $e_y'\notin M_1,M_2$. This means that 
$M_1\cap E(P_{xy}')=M_2\cap E(P_{xy}')=\{e\}$ contradicting that the sets are distinct.

Assume that $|M_1\cap E(P_{xy}')|$ is even. Since $M_1\cap E(P_{xy}')$ and $M_2\cap E(P_{xy}')|$ are nonempty, $|M_1\cap E(P_{xy}')|=|M_2\cap E(P_{xy}')|=2$.   
As we already observe in the previous case, $e_x\in M_1$ ($e_y\in M_1$, respectively) if and only if $e_x\in M_2$ ($e_y\in M_2$, respectively). 
In particular, if $M_1\cap E(P_{xy}')=\{e_x,e_y\}$, then  $M_2\cap E(P_{xy}')=\{e_x,e_y\}$; a contradiction.
Notice that if $e_x,e_y\notin M_1$, then $M_1\cap E(P_{xy}')=\{e_x',e_y'\}$. Then we obtain that  $M_2\cap E(P_{xy}')=\{e_x',e_y'\}$; a contradiction. This implies that 
either $e_x\in M_1,M_2$ and $e_y\notin M_1,M_2$ or, symmetrically,  $e_x\notin M_1,M_2$ and $e_y\in M_1,M_2$. In both cases $|M_1\cap \{e_x',e,e_y'\}|=|M_2\cap \{e_x',e,e_y'\}|=1$ and we obtain a contradiction using (iii)(d) in exactly the same way as in the previous case when .$|M_1\cap E(P_{xy}')|$ is odd.
This completes the proof of (iii).

Finally, we show (iv). Let $M'$ be a nonempty matching cut of $G'$. We show that $H$ has a nonempty matching cut $M$  that is equivalent to $M'$.
We construct $M$ as follows. First, we include in $M$ the edges of $M'\cap E(G[Y])$. Then for every long path $P_{xy}\in \mathcal{P}$, we do the following.
Denote by $e_x$ the $x$-edge, $e_x'$ the second $x$-edge, $e_y$ the $y$-edge, $e_y'$ the second $y$-edge, and $e$ the unique middle edge of $P_{xy}'$. 
\begin{itemize}
\item If $e_x\in M'$ ($e_y\in M'$, respectively), then include $e_x$ ($e_y$, respectively) in $M'$.
\item If $e_x,e_y\in M'$ and $|M'\cap E(P_{xy})|$ is odd, then include $e\in M$. If $e_x,e_y\in M'$ and $|M'\cap E(P_{xy})|$ is even, then 
no other edge of $P_{xy}'$ is included in $M$, that is, $M\cap E(P_{xy}')=\{e_x,e_y\}$.
\item If $e_x\in M'$ and $e_y\notin M'$, then
\begin{itemize}
\item if $|M'\cap E(P_{xy})|$ is odd, then no other edge of $P_{xy}'$ is included in $M$, that is, $M\cap E(P_{xy}')=\{e_x\}$,
\item if $|M'\cap E(P_{xy})|$ is even, then $e_y'$ is included in $M$ if $e_y'\in M'$ and $e$ is included in $M$ if $e_y'\notin M'$.
\end{itemize}
\item If $e_x\notin M'$ and $e_y\in M'$, then
\begin{itemize}
\item if $|M'\cap E(P_{xy})|$ is odd, then no other edge of $P_{xy}'$ is included in $M$, that is, $M\cap E(P_{xy}')=\{e_y\}$,
\item if $|M'\cap E(P_{xy})|$ is even, then $e_x'$ is included in $M$ if $e_x'\in M'$ and $e$ is included in $M$ if $e_x'\notin M'$.
\end{itemize}
\item If $e_x,e_y\notin M'$, then
\begin{itemize}
\item if $|M'\cap E(P_{xy})|$ is even, then $e_x'$ and $e_y'$ are included in $M$,
\item if $|M'\cap E(P_{xy})|$ is odd, then $e_x'$ is included in $M$ if $e_x'\in M'$ and $e_y'\notin M'$, 
$e_y'$ is included in $M$ if $e_x'\notin M'$ and $e_y'\in M'$, 
 and $e$ is included in $M$ if either  $e_x',e_y'\in M'$ or $e_x',e_y'\notin M'$.
\end{itemize}
\end{itemize}
It is straightforward to verify that $M$ is a matching of $H$ satisfying conditions (i)--(iii) of the definition of equivalency. Then using exactly the same argument as in the proof of (ii) we observe that  $M$ is a matching cut of $H$. This concludes the proof of (iv).
\end{proof}

Claim~\ref{cl:eq-fen} allows us to construct the solution-lifting algorithm for nonempty matching cuts of $H$ that outputs nonempty matching cuts from $\mathcal{M}_2$. For each nonempty matching cut $M$ of $H$, the algorithm lists the matching cuts $M'$ of $G'$ such that $M'$ is equivalent to $M$. Then for each $M'$, we extend $M'$ to matching cuts of $G$ by adding matchings of $F=G-E(G')$.  For this, we consider the algorithm $\textsc{EnumPath}(P_{x,y},A,B,C,h)$ that given a path $P_{xy}\in \mathcal{P}$, disjoint sets $A,B,C\subseteq E(P_{xy})$, and an integer $h\in\{0,1\}$, enumerates with polynomial delay all nonempty matchings $M$ of $P_{xy}$ such that $A\subseteq M$, $B\cap M=\emptyset$, either $C\subseteq M$ or $C\cap M=\emptyset$, and 
 $|M|\mod 2=h$. Such an algorithm exists by Observation~\ref{obs:fixed}. We also use the algorithm 
 $\textsc{EnumMatchF}(M)$ that, given a matching cut $M$ of $G'$, lists all matching cuts of~$G$ of the form $M\cup M'$, where $M'$ is a matching of $F$.  
 $\textsc{EnumMatchF}(M)$ is constructed as follows. Let $A$ be the set of edges of $F$ incident to the end-vertices of $F$ (recall that each connected component of $F$ contains at most one vertex of $V(G')$). Then we enumerate the matchings $M'$ of $F$ such that $M'\cap A=\emptyset$. This can be done with polynomial delay by Observation~\ref{obs:fixed}.

We use \textsc{EnumPath} and \textsc{EnumMatchF} as subroutines of the recursive branching algorithm 
\textsc{EnumEquivalent} (see Algorithm~\ref{alg:fen}) that, given a matching $M$ of $H$, takes as an input a matching $L$ of $G$ and $\mathcal{R}\subseteq \mathcal{P}$ and outputs the matching cuts $M'$ of $G$ such that 
(i)~$L\subseteq M'$, (ii) $M'$ is equivalent to $M$,  and (iii) the constructed matchings $M'$  differ only by some edges of the paths $P_{xy}\in\mathcal{R}$.
To initiate the computations, we construct the initial matching $L'$ of $G$ and the initial set of paths $\mathcal{R}'\subseteq \mathcal{P}$ as follows.
We define $\mathcal{R}'\subseteq\mathcal{P}$ to be the set of long paths $P_{xy}\subseteq \mathcal{P}$ such that $P_{xy}'\cap M\neq\emptyset$.  Then $L'\subseteq M$ is the set of edges of $M$ that are not in the paths of $\mathcal{R}'$. Recall that as an intermediate step, we enumerate nonempty matching cuts of $G'$ that are equivalent to $M$. Then it can be noted that to do this, we have to enumerate all possible extensions of $M$ to $M'$ satisfying condition (iii) of the equivalence definition.  Therefore, we call $\textsc{EnumEquivalent}(L',\mathcal{R}')$ to solve the enumeration problem. 

\begin{algorithm}[H]
\caption{$\textsc{EnumEquivalent}(L,\mathcal{R})$}\label{alg:fen}
\If{$\mathcal{R}=\emptyset$}{call  $\textsc{EnumMatchF}(M)$\;
return every matching cut $M'$ generated by the algorithm and {\bf quit}}
\ElseIf{$\mathcal{R}\neq\emptyset$}
{
select arbitrary $P_{xy}\in\mathcal{R}$\;
set $A:=\emptyset$; $B:=\emptyset$; $C:=\emptyset$; $h:=|M\cap E(P_{xy}')|\mod 2$\;
\lIf{$e_x\in M$}{set $A:=A\cup\{e_x\}$}
\lIf{$e_y\in M$}{set $A:=A\cup\{e_y\}$}
\lIf{$e_x'\in M$ and $e,e_y'\notin M$}{set $A:=A\cup\{e_x'\}$ and $B:=B\cup\{e_y'\}$}
\lIf{$e_y'\in M$ and $e,e_x'\notin M$}{set $A:=A\cup\{e_y'\}$ and $B:=B\cup\{e_x'\}$}
\lIf{$e\in M$ and $e_x',e_y'\notin M$}{set $C:=C\cup\{e_x',e_y'\}$}
call  $\textsc{EnumPath}(P_{x,y},A,B,C,h)$\;
\ForEach{nonempty matching $Z$ generated by $\textsc{EnumPath}(P_{x,y},A,B,C,h)$}
{$\textsc{EnumEquivalent}(L\cup Z, \mathcal{R}\setminus\{P_{xy}\})$
}
}
\end{algorithm}

Let us remark that when we call   $\textsc{EnumMatchF}(M)$ in line~(2), we immediately return each matching cut $M'$ produced by the algorithm. Similarly, when we call $\textsc{EnumPath}(P_{xy},A,B,C,h)$ in line (13), 
we immediately execute the loop in lines (14)--(16) 
for each generated matching~$Z$.

It can be seen from the description that \textsc{EnumEquivalent} is a backtracking enumeration algorithm. It picks $P_{xy}\in \mathcal{R}$ and produces nonempty matchings of $P_{xy}$. Notice that the sets of edges $A$, $B$, and $C$, the parity of the number of edges in a matching are assigned in lines (7)--(12) 
exactly as it is prescribed in condition (iii) of the equivalency definition.  Then the algorithm branches on all possible selections of the matchings of $P_{xy}$.The depth of the recursion is upper-bounded by $n$. This implies that $\textsc{EnumEquivalent}(L',\mathcal{R}')$ enumerates with polynomial delay all nonempty matching cuts $M\in \mathcal{M}_2$ 
such that $M'\cap E(G')$ is a nonempty matching cut of $G'$ equivalent to $M$ withe 
that are equivalent to $M'$. 

To summarize, recall that if $H$ is connected and has a vertex of degree one, we used the matching cut $\{e^*\}$ to list the matching cuts formed by the edges of $F=G-E(G')$. Clearly, $\{e^*\}$ is generated by  $\textsc{EnumEquivalent}(L',\mathcal{R}')$ for $L'$ and $\mathcal{R}'$ constructed for $M=\{e^*\}$. Therefore, 
we conclude that  the solution-lifting algorithm satisfies condition (ii$^*$) of the definition of 
a \pd{} enumeration kernel.
This finishes the proof of the theorem.
\end{proof}

Using Theorems~\ref{thm:fen-kern-min},  \ref{thm:fen-kern}, \ref{thm:general}, and~\ref{obs:natural}, we obtain the following corollary.

\begin{corollary}\label{thm:upper-fen}
The (minimal) matching cuts of an $n$-vertex graph $G$ can be enumerated  with 
$2^{\Oh(\fen(G))}\cdot n^{\Oh(1)}$ delay.
\end{corollary}

\section{Enumeration Kernels for the Parameterization by the Clique Partition Number}\label{sec:cp}
A partition $\{Q_1,\ldots,Q_k\}$ of the vertex set of a graph $G$ is said to be a \emph{clique partition} of $G$ if $Q_1,\ldots,Q_k$ are cliques. The \emph{clique partition} number of $G$, denoted by $\cp(G)$, is the minimum $k$ such that $G$ has a clique partition with $k$ cliques. Notice that the clique partition number of $G$ is exactly the chromatic number of its complement $\overline{G}$. In particular, this means that deciding whether $\cp(G)\leq k$ is \classNP-complete for any fixed $k\geq 3$~\cite{GareyJ79}. Therefore, throughout this section, we assume that the input graphs are given together with their clique decompositions. With this assumption, we show that the matching cut enumeration problems admit bijective kernels when parameterized by the clique partition number. Our result uses the following straightforward observation.
\begin{observation}\label{obs:clique}
Let $K$ be a clique of a graph $G$ such that $|K|\neq 2$. Then for every partition $\{A,B\}$ of $V(G)$ such that $E(A,B)$ is a matching, either $K\subseteq A$ or $K\subseteq B$.
\end{observation}
\begin{theorem}\label{thm:cp}
\probMC (\probMinMC, \probMaxMC, respectively) admits a bijective enumeration kernel with $\Oh(k^3)$ vertices if the input graph is given together with a clique partition of size at most $k$. 
\end{theorem}
\begin{proof} We show a bijective enumeration kernel for \probMC. 
Let $G$ be a graph and let $\mathcal{Q}$ be a clique partition of $G$ of size at most $k$. 
We apply a series of reduction rules.

First, we get rid of cliques of size two in $\mathcal{Q}$ to be able to use Observation~\ref{obs:clique}. We exhaustively apply the following rule.
\begin{reduction}\label{rule:size-two}
If $\mathcal{Q}$ contains a clique $Q=\{x,y\}$ of size two, then replace $Q$ by $Q_1=\{x\}$ and $Q_2=\{y\}$.
\end{reduction}
The rule does not affect $G$ and, therefore, it does not influence matching cuts of $G$. To simplify notation, we use $\mathcal{Q}$ for the obtained clique partition of $G$. Note that $|\mathcal{Q}|\leq 2k$. 

By the next rule, we unify cliques that cannot be separated by a matching cut. 
\begin{reduction}\label{rule:glue}
If $\mathcal{Q}$ contains distinct  cliques $Q_1$ and $Q_2$ such that $E(Q_1,Q_2)$ is nonempty and is not a matching, then make each vertex of $Q_1$ adjacent to every vertex of $Q_2$ and replace $Q_1,Q_2$ by the clique $Q=Q_1\cup Q_2$ in $\mathcal{Q}$. 
\end{reduction}
To see that the rule is safe, 
notice that if there are two distinct cliques $Q_1,Q_2\in\mathcal{Q}$ such $E(Q_1,Q_2)$ is nonempty and is not a matching, then for every partition $\{A,B\}$ of $V(G)$ such that $E(A,B)$ is a matching cut, either $Q_1,Q_2\subseteq A$ or $Q_1,Q_2\subseteq B$, because by Observation~\ref{obs:clique}, each clique of $\mathcal{Q}$ is either completely in $A$ or  in $B$. This means that if $G'$ is obtained from $G$ by the application of Rule~\ref{rule:glue}, then $M$ is a matching cut of $G'$ if and only if $M$ is a matching cut of $G$. 
Therefore, enumerating matching cuts of $G$ is equivalent to enumerating matching cuts of $G'$.

We apply Rule~\ref{rule:glue} exhaustively. Denote by $\hat{G}$ the obtained graph and by $\hat{\mathcal{Q}}$ the obtained clique partition. Notice that the rule is never applied for a pair of cliques of size one and, therefore, $\hat{\mathcal{Q}}$ does not contain cliques of size two.
 
In the next step, we use the following marking procedure to label some vertices of $\hat{G}$. 

\paragraph*{Marking procedure.}
\begin{itemize}
\item[(i)] For each pair $\{Q_1,Q_2\}$ of distinct cliques of $\hat{\mathcal{Q}}$ such that $E(Q_1,Q_2)\neq\emptyset$, select arbitrarily an edge $uv\in E(\hat{G})$ such that $u\in Q_1$ and $v\in Q_2$, and \emph{mark} $u$ and $v$.
\item[(ii)] For each triple  $Q,Q_1,Q_2$ of distinct cliques of $\hat{\mathcal{Q}}$ such that  some  vertex of $Q$ is adjacent to a vertex in $Q_1$ and to a vertex of $Q_2$, select arbitrarily $u\in Q$, $x\in Q_1$ and $y\in Q_2$ such that $ux,uy\in E(\hat{G})$, and then \emph{mark} $u,x,y$.
\end{itemize}
Notice that a vertex may be marked several times.

Our final rule reduces the size of the graph by deleting some unmarked vertices.
\begin{reduction}\label{rule:reduce}
For every clique $Q\in\hat{\mathcal{Q}}$ of size at least three, consider the set $X$ of unmarked vertices of $Q$ and delete arbitrary $\min\{|X|,|Q|-3\}$ vertices of $X$.
\end{reduction}
Denote by $\tilde{G}$ the obtained graph and denote by $\tilde{\mathcal{Q}}$ the corresponding clique partition of $\tilde{G}$ such that every clique of $\tilde{\mathcal{Q}}$ is obtained by the deletion of unmarked vertices from a clique of $\hat{\mathcal{Q}}$. Our kernelization algorithm return $\tilde{G}$ together with $\tilde{\mathcal{Q}}$ (recall that by our convention each instance should be supplied with a clique partition of the input graph). 

To upper bound the size of the obtained kernel, we show the following claim.
\begin{claim}\label{cl:size}
The graph $\tilde{G}$ has at most $4k(3k^2-2k+1)$ vertices.
\end{claim}
\begin{proof}[Proof of Claim~\ref{cl:size}]
Let $r=|\hat{\mathcal{Q}}|\leq 2k$. In Step (i) of the marking procedure, we consider $\binom{r}{2}$ pairs of cliques of $\hat{\mathcal{Q}}$ and mark at most $2\binom{r}{2}\leq 2k(2k-1)$ vertices in total. In Step (ii), we consider $r$ cliques $Q$, and for each $Q$, we consider $\binom{r-1}{2}$ pairs $\{Q_1,Q_2\}$. This implies that we mark at most 
$3r\binom{r-1}{2}\leq 3k(2k-1)(2k-2)$ vertices in this step. Hence, the total number of marked vertices is at most 
$3k(2k-1)(2k-2)+2k(2k-1)=2k(2k-1)(3k-2)$. Since each clique of $\tilde{\mathcal{Q}}$ contains at most three unmarked vertices by Rule~\ref{rule:reduce}, 
$|V(\tilde{G})|\leq 2k(2k-1)(3k-2)+3r\leq 2k(2k-1)(3k-2)+6k=4k(3k^2-2k+1)$.
\end{proof}

This completes the description of the kernelization algorithm. Now we show a bijection between matching cuts of $G$ and $\tilde{G}$, and construct our solution-lifting algorithm. Note that we already established that $M$ is a matching cut of $\hat{G}$ if and only if $M$ is a matching cut of $G$. 
Therefore, it is sufficient to construct a bijective mapping of the matching cuts of $\tilde{G}$ to the matching cuts of $\hat{G}$.

Let $\hat{\mathcal{Q}}=\{\hat{Q}_1,\ldots,\hat{Q}_r\}$ and $\tilde{\mathcal{Q}}=\{\tilde{Q}_1,\ldots,\tilde{Q}_r\}$, where $\tilde{Q}_i\subseteq \hat{Q}_i$ for $i\in\{1,\ldots,r\}$. 
Notice that $\hat{\mathcal{Q}}$ and  $\tilde{\mathcal{Q}}$ have no cliques of size two. Hence, we can use Observation~\ref{obs:clique}. We show the following claim.
\begin{claim}\label{cl:mapping}
For every partition $\{I,J\}$ of $\{1,\ldots,r\}$, $\hat{M}=E(\cup_{i\in I}\hat{Q}_i,\cup_{j\in J}\hat{Q}_j)$ is a matching cut of $\hat{G}$ if and only if $\tilde{M}=E(\cup_{i\in I}\tilde{Q}_i,\cup_{j\in J}\tilde{Q}_j)$ is a matching cut of $\tilde{G}$. Moreover, $\hat{M}=\emptyset$ if and only if $\tilde{M}=\emptyset$. 
\end{claim} 
\begin{proof}[Proof of Claim~\ref{cl:mapping}]
If $\hat{M}=E(\cup_{i\in I}\hat{Q}_i,\cup_{j\in J}\hat{Q}_j)$ is a matching cut of $\hat{G}$, then $\tilde{M}$ is a matching cut of $\tilde{G}$, because $\tilde{Q}_i\subseteq \hat{Q}_i$ for $i\in\{1,\ldots,r\}$.

For the opposite direction, assume that $\tilde{M}=E(\cup_{i\in I}\tilde{Q}_i,\cup_{j\in J}\tilde{Q}_j)$ is a matching cut of $\tilde{G}$. For the sake of contradiction, suppose that 
$\hat{M}=E(\cup_{i\in I}\hat{Q}_i,\cup_{j\in J}\hat{Q}_j)$ is not a matching cut of $\hat{G}$. 
This means that 
there is a vertex that is incident to at least two edges of $\hat{M}$. By symmetry, we can assume without loss of generality that there is $h\in I$ and $u\in \hat{Q}_h$ such that $u$ is adjacent to distinct vertices $x$ and $y$, where $x\in \hat{Q}_i$ and $y\in\hat{Q}_j$ for some $i,j\in J$. 

Suppose that $i=j$, that is, $x$ and $y$ are in the same clique $\hat{Q}_i$. Then, however, $E(\hat{Q}_h,\hat{Q}_i)$ is not a matching and we would be able to apply Rule~\ref{rule:glue}. Since Rule~\ref{rule:glue} was applied exhaustively to obtain $\hat{G}$ and $\hat{\mathcal{Q}}$, this cannot happen. We conclude that $i\neq j$. 

By Step (ii) of the marking procedure, there is $u'\in \hat{Q}_h$, $x\in \hat{Q}_i$ and $y'\in \hat{Q}_j$ such that $u'x',u'y'\in E(\hat{G})$ and the vertices $u',x',y'$ are marked. This means that $u'\in\tilde{Q}_h$, $x'\in\tilde{Q}_i$ and $y'\in\tilde{Q}_j$. Then $u'x',u'y'\in\tilde{M}$ and $\tilde{M}$ is not a matching; a contradiction. The obtained contradiction concludes the proof.

Finally, we show that $\hat{M}=\emptyset$ if and only if $\tilde{M}=\emptyset$. Clearly, if $\hat{M}=\emptyset$, then  $\tilde{M}=\emptyset$. Suppose that  $\hat{M}\neq\emptyset$. 
Then 
there are $i\in I$ and $j\in J$ such that $uv\in\hat{M}$ for some $u\in \hat{Q}_i$  and $v\in \hat{Q}_j$. By Step (i) of the marking procedure, there are $u'\in \hat{Q_i}$ and $v'\in \hat{Q}_j$ such that $u'v'\in E(\hat{G})$ and $u',v'$ are marked. Then $u'v'\in E(\tilde{G})$ and $u'v'\in\tilde{M}$ by the definition of $\tilde{M}$. Hence, $\tilde{M}\neq\emptyset$. 
\end{proof}
Using Claim~\ref{cl:mapping}, we are able to describe our solution-lifting algorithm. Let $\tilde{M}$ be a matching cut of $\tilde{G}$. Let $\{A,B\}$ be a partition of $V(\tilde{G})$ such that $\tilde{M}=E(A,B)$. By Observation~\ref{obs:clique}, there is a partition $\{I,J\}$ of $\{1,\ldots,r\}$ such that $A=\cup_{i\in I}\tilde{Q}_i$ and $B=\cup_{j\in J}\tilde{Q}_j$. The solution-lifting algorithm outputs $\hat{M}=E(\cup_{i\in I}\hat{Q}_i,\cup_{j\in J}\hat{Q}_j)$. 

To see correctness, note that $\hat{M}=E(\cup_{i\in I}\hat{Q}_i,\cup_{j\in J}\hat{Q}_j)$ is a matching cut of $\hat{G}$ by Claim~\ref{cl:mapping}. Moreover, for distinct matching cuts $\tilde{M}_1$ and $\tilde{M}_2$ of $\tilde{G}$, the constructed matching cuts $\hat{M}_1$ and $\hat{M}_2$, respectively, of $\hat{G}$ are distinct, that is, the matching cuts of $\tilde{G}$  are mapped to the matching cuts of $\hat{G}$ injectively. Finally, to show that the mapping is bijective, consider a matching cut $\hat{M}$ of~$\hat{G}$. Let $\{A,B\}$ be a partition of $V(\hat{G})$ with $\hat{M}=E(A,B)$. By Observation~\ref{obs:clique}, there is a partition $\{I,J\}$ of $\{1,\ldots,r\}$ such that $A=\cup_{i\in I}\hat{Q}_i$ and $B=\cup_{j\in J}\hat{Q}_j$. Then 
$\tilde{M}=E(\cup_{i\in I}\tilde{Q}_i,\cup_{j\in J}\tilde{Q}_j)$ is a matching cut of $\tilde{G}$ by Claim~\ref{cl:mapping}, and it remains to observe that 
$\hat{M}$ is constructed by the solution-lifting algorithm from $\tilde{M}$.

It is straightforward to see that both kernelization and solution-lifting algorithms are polynomial.
This concludes the construction of our bijective enumeration kernel for \probMC. 

For \probMinMC and \probMaxMC, the kernels are exactly the same. To see it, note that our bijective mapping of the matching cuts of $\tilde{G}$ to the matching cuts of $\hat{G}$ respects the inclusion relation. Namely, if $\tilde{M}_1$ and $\tilde{M}_2$ are matching cuts of $\tilde{G}$ and $\tilde{M}_1\subseteq\tilde{M}_2$, then $\tilde{M}_1$ and $\tilde{M}_2$ are mapped to the matching cuts $\hat{M}_1$ and $\hat{M}_2$, respectively, of $\hat{G}$ such that $\hat{M}_1\subseteq \hat{M}_2$. This implies that the solution-lifting algorithm outputs a minimal (maximal, respectively) matching cut of $\hat{G}$ for every minimal (maximal, respectively) matching cut of $\tilde{G}$. Moreover, every minimal (maximal, respectively) matching cuts of $\hat{G}$ can be obtained form a minimal (maximal, respectively) matching cut of $\tilde{G}$.
\end{proof}

\section{Conclusion}\label{sec:concl}  
We initiated the systematic study of enumeration kernelization for several variants of
the matching cut problem. We obtained \fp{} (\pd{}) enumeration kernels for the
parameterizations by the vertex cover number, twin-cover number, neighborhood diversity,
modular width, and feedback edge number. Since the solution-lifting algorithms are simple
branching algorithms, these kernels give a condensed view of the solution sets which may be interesting in applications where one may want to inspect all solutions manually. Restricting to polynomial-time and polynomial-delay solution-lifting algorithms seems helpful in the sense that they will usually be easier to understand.

There are many topics for further research in enumeration kernelization. For \textsc{Matching Cut}, it would be interesting to investigate other structural parameters, like the feedback vertex number (see~\cite{CyganFKLMPPS15} for the definition). More generally, the area of enumeration kernelization seems still somewhat unexplored. It would be interesting to see applications of the various kernel types to other enumeration problems. For this, it seems to be important to develop general tools for enumeration kernelizations. For example, is it possible to establish a framework for enumeration kernelization lower bounds similar to the techniques used for classical kernels~\cite{BodlaenderDFH09,BodlaenderJK14} (see also~\cite{CyganFKLMPPS15,FominLSZ19})?

Concerning the counting and enumeration of matching cuts, we also proved the upper bound $F(n+1)-1$ for the maximum number of matching cuts of an $n$-vertex graph and showed that the bound is tight. What can be said about the maximum number of minimal and maximal matching cuts? 
It is not clear whether our lower bounds given in Propositions~\ref{prop:lower-max} and \ref{prop:lower-min} are tight. Finally, it seems promising to study  enumeration kernels for \textsc{$d$-Cut}~\cite{GomesS19}, a generalization of \textsc{Matching Cut} that has recently received some attention.


\end{document}